\documentclass[onecolumn,superscriptaddress,longbibliography]{revtex4-2}
\usepackage[ruled,vlined]{algorithm2e}

\usepackage[margin=1in]{geometry}
\usepackage{graphicx}
\usepackage{multirow}
\usepackage{amsmath,amsthm}
\usepackage{amssymb,amsfonts}
\usepackage{soul}
\usepackage{verbatim}
\usepackage{bm}
\usepackage[colorlinks=true, linkcolor=blue, citecolor=blue, urlcolor=blue]{hyperref}
\usepackage[table, dvipsnames]{xcolor}
\usepackage{tikz}
\usepackage{xspace}
\usepackage{import}
\usepackage[caption=false]{subfig}
\usepackage{bbold}
\usepackage{qcircuit}
\usepackage{chemformula}
\usepackage{adjustbox}
\usepackage[ruled]{algorithm2e} 
\usepackage{newfloat,algcompatible} 
\usepackage{parskip}
\newcommand{\ket}[1]{|#1\rangle}
\newcommand{\bra}[1]{\langle#1|}

\newcommand{\lb}{\left(}
\newcommand{\rb}{\right)}
\newcommand{\braket}[2]{\langle#1|#2\rangle}

\newcommand{\norm}[1]{\left\lVert#1\right\rVert}

\makeatletter
\newcommand\footnoteref[1]{\protected@xdef\@thefnmark{\ref{#1}}\@footnotemark}
\makeatother

\newcommand{\ketbra}[2]{|#1\rangle\!\langle #2|}

\newcommand{\thm}[1]{\hyperref[thm:#1]{Theorem~\ref*{thm:#1}}}
\newcommand{\defn}[1]{\hyperref[defn:#1]{Definition~\ref*{defn:#1}}}
\newcommand{\lem}[1]{\hyperref[lem:#1]{Lemma~\ref*{lem:#1}}}
\newcommand{\prop}[1]{\hyperref[prop:#1]{Proposition~\ref*{prop:#1}}}
\newcommand{\fig}[1]{\hyperref[fig:#1]{Figure~\ref*{fig:#1}}}
\newcommand{\tab}[1]{\hyperref[tab:#1]{Table~\ref*{tab:#1}}}
\renewcommand{\sec}[1]{\hyperref[sec:#1]{Section~\ref*{sec:#1}}}
\newcommand{\app}[1]{\hyperref[app:#1]{Appendix~\ref*{app:#1}}}
\newcommand{\cor}[1]{\hyperref[cor:#1]{Corollary~\ref*{cor:#1}}}
\newcommand{\obs}[1]{\hyperref[obs:#1]{Observation~\ref*{obs:#1}}}

\newcommand{\paragraphsentence}[1]{}
\renewcommand{\paragraphsentence}[1]{\emph{#1}}

\newcommand{\One}{\mathbb{1}}

\usepackage{amsthm,amsfonts}
\usepackage{tcolorbox}
\newtheorem{theorem}{Theorem}
\newtheorem{lemma}[theorem]{Lemma}
\newtheorem{definition}[theorem]{Definition}
\newtheorem{corollary}[theorem]{Corollary}

\newcommand{\UofT}{\affiliation{
Department of Computer Science, University of Toronto, Canada}}

\newcommand{\UofTP}{\affiliation{
Department of Physics, University of Toronto, Canada}}

\newcommand{\PNNL}{\affiliation{Pacific Northwest National Laboratory, Richland WA, USA}}

\newcommand{\CIFAR}{\affiliation{Canadian Institute for Advanced Studies, Toronto, Canada}}

\newcommand{\BIQuantum}{\affiliation{
Quantum Lab, Boehringer Ingelheim, 55218 Ingelheim am Rhein, Germany}}

\begin{document}

\title{Amplified Amplitude Estimation: Exploiting Prior Knowledge to Improve Estimates of Expectation Values}

\author{Sophia Simon}
\email{sophia.simon@mail.utoronto.ca}
\UofTP

\author{Matthias Degroote}
\email{matthias.degroote@boehringer-ingelheim.com}
\BIQuantum

\author{Nikolaj Moll}
\BIQuantum

\author{Raffaele Santagati}
\BIQuantum

\author{Michael Streif}
\BIQuantum

\author{Nathan Wiebe}
\email{nathan.wiebe@utoronto.ca}
\UofT
\PNNL
\CIFAR

\begin{abstract}
We provide a method for estimating the expectation value of an operator that can utilize prior knowledge to accelerate the learning process on a quantum computer. Specifically, suppose we have an operator that can be expressed as a concise sum of projectors whose expectation values we know a priori to be $O\lb \epsilon \rb$. In that case, we can estimate the expectation value of the entire operator within error $\epsilon$ using a number of quantum operations that scales as $O(1/\sqrt{\epsilon})$. We then show how this can be used to reduce the cost of learning a potential energy surface in quantum chemistry applications by exploiting information gained from the energy at nearby points. Furthermore, we show, using Newton-Cotes methods, how these ideas can be exploited to learn the energy via integration of derivatives that we can estimate using a priori knowledge. This allows us to reduce the cost of energy estimation if the block-encodings of directional derivative operators have a smaller normalization constant than the Hamiltonian of the system.
\end{abstract}

\maketitle

\section{Introduction}

Quantum simulation algorithms have progressed to the point where their asymptotic worst-case computational complexity is essentially optimal \cite{berry2001optimal, Low2017QSP, Low2019hamiltonian, Gilyen2019qsvt}. Despite this and improved constant factor analyses \cite{Su2021first_quant, Steudtner2023observables}, current resource estimates for industrially relevant applications such as FeMoco~\cite{reiher2017elucidating,von2021quantum,lee2021even} indicate that solving these problems with existing algorithms will likely remain out of reach for early fault-tolerant quantum computers let alone NISQ-era quantum computers.  This suggests that further improvements may be needed to achieve a practical advantage.

Crucially though, classical algorithms for quantum chemistry algorithms try to use prior knowledge as much as possible, while quantum algorithms tend to use such information less frequently. Cases where prior knowledge is used in quantum computing include state learning algorithms such as compressed sensing~\cite{flammia2012quantum}, ``warm starts" in quantum machine learning~\cite{wiebe2016quantum}, quantum machine learning with inductive biases \cite{Kubler2021inductive_bias, Larocca2022QML, Cerezo2022QML, Ragone2023GQML} and even the phase estimation protocol in Shor's factoring algorithm~\cite{shor1994algorithms}. These examples exploit a special structure that is known in the problem to surpass the performance of algorithms that do not possess this knowledge.  Quantum simulation often uses prior knowledge in limited ways.   In particular, in applications such as chemistry, we can often use classically tractable methods such as Hartree-Fock (HF) or Density Functional Theory (DFT) to begin with an initial estimate of quantities of interest such as the ground state energy~\cite{sholl2022density,he2019density}.  Existing quantum algorithms seldom depend strongly on this knowledge apart from its use in providing an ansatz that has large overlap with the target eigenstate that we may wish to prepare (for problems that involve simulation of ground state energies or related quantities) \cite{Ge2019groundstate, Lin2020ground_state}.  This is a significant drawback as in the most extreme case, perfect classical knowledge about the quantity one wishes to estimate eliminates the need for a quantum computing simulation, while in more realistic cases, it might drastically reduce the cost of the simulation.  Thus providing a new approach to simulation that properly utilizes this prior knowledge is a promising avenue if we wish to demonstrate a useful quantum advantage on intermediate-term quantum computers which may not be able to execute the billions of gate operations needed by existing algorithms~\cite{von2021quantum,lee2021even,Obrien2022,Goings2022,Steudtner2023observables,Cortes2023, Simon2023Liouvillian}.

In this work, we take a step towards solving this problem by introducing an approach to learning expectation values that we call amplified amplitude estimation (AAE). The idea behind this technique is to use amplitude amplification in concert with amplitude estimation to estimate probabilities that are \emph{a priori} known to be small ($O(\epsilon)$) within error $\epsilon$ using a number of queries that scales only as $O(1/\sqrt{\epsilon})$.
We show how AAE can be used to estimate expectation values of operators which do not satisfy the smallness assumption directly, assuming that we have some prior knowledge of the expectation values we wish to estimate. More specifically, we use linear combinations of unitaries~\cite{childs2012hamiltonian} together with techniques for implementing matrix square roots of positive semi-definite (PSD) operators to reduce the problem of estimating generic expectation values to the case of estimating small probabilities.
This allows us to extend the $O(1/\sqrt{\epsilon})$ scaling beyond the case of small probability estimation.  

We provide several applications for these ideas.  First, we apply the idea to expectation value estimation of one-body operators for chemistry.  Specifically, we show that this approach can yield a polynomial advantage if the uncertainty provided is small relative to the mean value of the coefficients of the operator that we wish to measure.  We then discuss how to leverage knowledge of the expectation value for a given molecular configuration to learn the expectation value for a nearby configuration at low cost, assuming that an appropriate smoothness condition on the coefficients in the decomposition of the one-body operator is satisfied.

Our second application involves estimating energy differences using Newton-Cotes formulas to integrate the energy derivative.  We show, under similar assumptions to the one-body operator case, that we can use this approach to take a set of configurations for a molecular system and interpolate from them in order to estimate the energy at nearby points at a lower cost than you would ordinarily be able to do so using standard approaches to estimating the energy of the system such as phase estimation.

\section{Background}

It is a common folk theorem in quantum metrology and quantum computing that classical strategies can be used to estimate a quantity within error $\epsilon$ (with high probability) using $O(1/\epsilon^2)$ applications of a preparation unitary whereas coherent methods such as phase estimation can learn the same quantity using $O(1/\epsilon)$ applications.  The latter scaling is often referred to as Heisenberg limited scaling, whereas the former is called shot noise limited scaling.  It is perhaps less well appreciated that small or large probabilities can actually be estimated using quadratically better scaling than would otherwise be expected.  Specifically, imagine that we have a state of the form $\sqrt{1-\epsilon} \ket{0} + \sqrt{\epsilon}\ket{1}$ and that our goal is to learn the probability of $\ket{1}$ within error $O(\epsilon)$.  Using the sample mean as an unbiased estimator of the probability, we have that the number of samples scales as 
\begin{align}
    N_{\rm samp} \in O\left(\frac{\sigma^2}{\epsilon^2} \right) = O\left(\frac{\epsilon(1-\epsilon)}{\epsilon^2} \right) = O\left(\frac{1}{\epsilon} \right),
\end{align}
where $\sigma^2$ denotes the variance. This represents a quadratic advantage compared to the shot noise-limited scaling. The guarantee that the amplitude of $\ket{1}$ is small, i.e. $O \lb \sqrt{\epsilon} \rb$, is essential in this scaling and can be considered prior knowledge. In this work, we show that the scaling for estimating the expectation value of an operator with prior knowledge can be further improved to $O\left(\frac{1}{\sqrt{\epsilon}} \right)$ through amplitude amplification together with amplitude estimation.

The idea behind standard amplitude estimation~\cite{brassard2002quantum} is intimately linked to that of amplitude amplification.  Specifically, amplitude amplification is a way of transducing a probability into a phase.  After converting the probability into a phase, the quantum phase estimation algorithm can be employed to learn the resultant phase.  If we wish to compute the probability $P =|\braket{0}{\psi}|^2$ for some quantum state $\ket{\psi}$, then we can estimate it by constructing a walk operator $W = -(\openone - 2 \ketbra{\psi}{\psi})(\openone -2\ketbra{0}{0})$.  The eigenvalues of $W$ that are supported by $\ket{\psi}$ are then of the form $e^{\pm i \theta}$ where $\theta = \sin^{-1}(\sqrt{P})$. Thus we can learn a probability $P$ by performing quantum phase estimation on $W$ to yield $\theta$ from which we can estimate $P= \sin^2(\theta)$.
The standard statement of amplitude estimation's performance is given below.

\begin{theorem}[Based on Theorem 12 in Brassard et al.~\cite{brassard2002quantum}]
Let $O_{\psi} : \ket{0} \mapsto \ket{\psi}$ and $O_{\phi}: \ket{0} \mapsto \ket{\phi}$ be invertible oracles and let $R= \openone - 2 O_\psi \ket{0}\!\bra{0}O_\psi^\dagger$, be the reflection with respect to $\ket{\psi}$.  There exists an algorithm that can estimate, for any state $\ket{\phi}$, $\bra{\phi} ( \openone- R)/2 \ket{\phi}$ within error $\epsilon$ and with a probability of success greater than $8/\pi^2$ while using a number of queries to  $O_\psi$ and $O_\phi$  that is in $O(1/\epsilon)$.\label{Th:Brassard2002}
\end{theorem}

This result is known to be optimal~\cite{berry2001optimal}.  It is impossible to achieve better than an $\Omega(1/\epsilon)$ scaling for the phase estimation problem without making further assumptions. A key assumption behind this result is that no prior knowledge is given about the value of the phase in question. While the estimation of small probabilities provides superior scaling, as shown above, its advantage is balanced by the assumptions of small probability and an equally small tolerance for the estimation error. It remains a question whether it is possible to find a method that, under certain assumptions, can provide the $O(1/\sqrt{\epsilon})$ scaling without requiring that the expectation value that we wish to estimate is small. 
We address this question with a new technique called amplified amplitude estimation, described in greater detail below. The general idea is as follows: if we have a good estimate of the expectation value that we wish to improve, we can transform the question into the problem of estimating the difference between the prior value we have and the true value we wish to know. In this way, we have transformed the problem into the estimation of a small probability.

\section{Amplified Amplitude Estimation}

Imagine we know the probability of a measurement outcome within some small error. The question is how this knowledge reduces the cost of probability estimation. In this section, we provide a way to leverage this knowledge by applying amplitude amplification on a small probability to boost it to a much larger probability if our prior knowledge rules out the possibility of over-rotation.  
We will later show how this can be applied to solve the problem of estimating the expectation values of operators.

We reframe the problem of estimating the overlap between two states of Theorem~\ref{Th:Brassard2002} in terms of the expectation value of a projector, $\bra{\psi}\Pi \ket{\psi}$, i.e., Brassard et al.\ estimate $\langle \psi | \Pi_{\rm good} |\psi \rangle =  \bra{\psi}(\openone- R)/2  \ket{\psi}$ where $\Pi_{\rm good}$ is a projector onto a ``good'' subspace and $R$ is a reflection operator~\cite{brassard2002quantum}. For the case where the ``good'' subspace consists of only a projector onto a state $\ket{\phi}$, the resulting probability is simply $|\braket{\psi}{\phi}|^2$. We can estimate the good probability within error $\epsilon$ using their approach in $O(1/\epsilon)$ queries to the underlying state preparation unitaries, which is known to be optimal in the worst-case scenario.  

A fundamental difference compared to Brassard et al.~\cite{brassard2002quantum} is that here we assume having access to prior knowledge in the form of an upper bound, $P_0$, on the quantity we want to estimate, $\bra{\psi}\Pi\ket{\psi}$. The new quantity we want to estimate is then the deviation from this upper bound, $\delta := \bra{\psi}\Pi\ket{\psi} - P_0$. This assumption of prior knowledge precludes the worst-case bounds discussed previously, which allows us, under certain restrictions, to show better scaling than those results allow. More specifically, we use AAE to estimate $\delta$, which then also yields an estimate of $\bra{\psi}\Pi\ket{\psi}$.
The number of queries AAE needs to estimate such a ``good'' or marked probability within error $\epsilon$ under these assumptions is given below.

\begin{lemma}[Amplified Amplitude Estimation] 
\label{lem:estimate_small_prob}
Let $\Pi$ be a projector and assume we have access to a unitary oracle $O_{\psi}$ such that $O_{\psi} \ket{0} = \ket{\psi}$ and a reflection oracle $R_\Pi$ such that $\Pi := (\openone-R_{\Pi})/2$. We assume that $\bra{\psi} \Pi \ket{\psi} = P_0 + \delta$, where $P_0=\sin^2(\pi/(2(2\mu+1)))$ is a known upper bound on the quantity $\bra{\psi} \Pi \ket{\psi}$ which we will interpret as a success probability. Furthermore, $\mu\ge 1$ is an integer and $\delta \in [-P_0, 0)$ is an unknown negative number. Then there exists a quantum algorithm that can estimate $\delta$, and hence $\bra{\psi} \Pi \ket{\psi}$, within error $\epsilon \in O(P_0^2/|\delta|) \subseteq o(1)$ and with failure probability at most $\delta'$ using 
$$
    N_\textrm{queries} \in O\left(\frac{P_0^{1/2} \log(1/\delta')}{\epsilon}\left( \frac{P_0}{|\delta|}\right) \right)
$$
queries to $R_{\Pi}$ and $O_{\psi}$.
\end{lemma}

\begin{proof}
By assumption, we have that $P_0$ is a probability such that for some positive integer $\mu\ge 1$
\begin{equation}
    \sin^{-1}(\sqrt{P_0}) = \frac{\pi}{2(2\mu+1)}
\label{eq:P0def},
\end{equation} 
which is to say that the angle in the rotation carried out by a single application of amplitude amplification is at most $\frac{\pi}{2(2\mu + 1)}$.
Further, we know that there exists a projector $\Pi$ and a negative number $\delta$ such that
\begin{equation}
    \bra{\psi} \Pi \ket{\psi} = P_0 + \delta = \sin^2\left(\frac{\pi}{2(2\mu+1)}\right) + \delta.
\label{eq:probdef}
\end{equation}
Next, note that the reflection operator $R_{\psi} := (\openone -2 \ketbra{\psi}{\psi})$ can be implemented using the unitary oracle $O_\psi$:
\begin{equation}
R_{\psi} =  \openone-2O_\psi\ketbra{0}{0}O_\psi^\dagger.
\end{equation}

This allows us to implement amplitude amplification through the walk operator
\begin{equation}
    W := - R_{\psi} R_{\Pi}= -  O_\psi \lb \openone-2 \ketbra{0}{0} \rb O_{\psi}^\dagger \lb \openone-2 \Pi \rb 
\end{equation}
using $O(1)$ queries to our oracles $O_\psi$ and $R_{\Pi}$~\cite{brassard2002quantum} .

Next, consider the boosted probability that would be found by applying $W^{\mu}$, corresponding to $\mu$ applications of $W$, to the state $\ket{\psi}$ to amplify the probability given in Eq.~\eqref{eq:probdef}. Using this amplified oracle is equivalent to constructing a new walk operator $W'$ such that
\begin{equation}
     W' := - \lb \openone - 2 W^\mu O_\psi \ketbra{0}{0}O_\psi^\dag {W^\mu}^\dag \rb \lb \openone -2\Pi \rb.  
\end{equation} 
The idea then is to perform amplitude estimation with the new walk operator $W'$.
This operator, from the analysis of~\cite{brassard2002quantum}, has eigenvalues (within the two-dimensional subspace in question) of the form $\exp(\pm i \arcsin{\sqrt{P_1}})$ where
\begin{equation}
    P_1 = \sin^2\left((2\mu +1) \sin^{-1}(\sqrt{P_0 +\delta}) \right)=\sin^2\left(\frac{\pi}{2}\left(\frac{\sin^{-1}(\sqrt{P_0 +\delta})}{\sin^{-1}(\sqrt{P_0})} \right)\right).
\label{eq:P1def}
\end{equation}

If we are given $\hat{P}_1$, an $\epsilon'$-accurate estimate of the probability, then we can solve for our estimate $\hat{\delta}$ of $\delta$ as a function of $\epsilon':= P_1-\hat{P}_1 $. This expression and a series expansion for the quantity in terms of small $\epsilon'$ and $\delta$ can be found using elementary algebra and is given below: 
\begin{align}
    \hat\delta &= \hat{\bra{\psi}\Pi\ket{\psi}} -P_0 \nonumber\\
    &= \sin^2\left(\frac{2\sin^{-1}(\sqrt{\hat{P}_1})\sin^{-1}(\sqrt{{P}_0})}{\pi}\right)-P_0\nonumber\\
    &= \sin^2\left(\frac{2\sin^{-1}(\sqrt{{P}_1-\epsilon'})\sin^{-1}(\sqrt{{P}_0})}{\pi}\right)-P_0\nonumber\\
    &=\delta -\frac{8P_0(1-P_0) \arcsin^2(\sqrt{P_0})\epsilon'}{\pi^2 \delta} + O(\epsilon' (1+\epsilon'/\delta)).
\label{eq:deltaerr}
\end{align}
Here we require $P_0 < 1$ for the leading order term in the asymptotic series to have this scaling since for $P_0=1$ the arcsin and sine trivially invert each other, making this form inappropriate for a series expansion along with causing the leading order term in the previous expansion to vanish. Note though that the condition $P_0 < 1$ follows directly from Eq.~\eqref{eq:P0def} if $\mu \geq 1$ since we then have that $P_0\leq \arcsin(\pi/6)^2= 1/4$. For this reason, we assume that $\mu \ge 1$, which then automatically ensures that $P_0 <1$.
From~\eqref{eq:deltaerr} we can obtain the error in the estimated value of $\delta$:
\begin{align}
    |\hat{\delta} - \delta| &\in O\left(\frac{P_0(1-P_0) \arcsin^2(\sqrt{P_0})\epsilon'}{|\delta|} \right).
\label{eq:est_delta}
\end{align}
Note that we neglected the $O(\epsilon' (1+\epsilon'/\delta))$ term from Eq.~\eqref{eq:deltaerr} since $\epsilon' < \frac{\epsilon'}{\delta}$.

Next, to understand the asymptotic behavior of this expression, it is useful to consider the case where $P_0 \ll 1$ such that $\arcsin(\sqrt{P_0}) \in \Theta(\sqrt{P_0})$ and similarly, $1-P_0 \in \Theta(1)$.  This leads to the conclusion that
\begin{equation}
|\hat{\delta} - \delta| \in O\left(\frac{P_0^2\epsilon'}{|\delta|} \right).
\end{equation}
Thus, if we wish to estimate $\delta$ within error $\epsilon$, it suffices to choose
\begin{equation}
\epsilon' \in \Theta\left(\frac{\epsilon |\delta|}{P_0^2} \right),
\label{eps_bound}
\end{equation}
which implies that if we wish to perform amplitude estimation within error $\epsilon$ and probability of failure smaller than $\delta'$, we will need a number of queries to the boosted walk operator $W'$ that obeys~\cite{brassard2002quantum, Ge2019groundstate} 
\begin{equation}
N_{W'} \in O\left(\frac{P_0^2\log(1/\delta')}{\epsilon |\delta|} \right).
\end{equation}
The factor of $\log(1/\delta')$ comes from the need to repeat the amplitude estimation procedure $\log(1/\delta')$ times to ensure that the failure probability is at most $\delta'$.
Further, taking Eq.~\eqref{eps_bound} into account, the smallness assumptions leading to Eq.~\eqref{eq:deltaerr} are satisfied if we pick $\epsilon \in O(P_0^2/|\delta|) \subseteq o(1)$.
As there are $O(\mu) = O(1/\sqrt{P_0})$ applications of $W$ needed for each application of the walk operator $W'$, we then find that
\begin{equation}
N_\textrm{queries} = N_W \in O(\mu N_{W'}) = O\left(\frac{P_0^{1/2} \log(1/\delta')}{\epsilon}\left( \frac{P_0}{|\delta|}\right) \right).
\end{equation}
\end{proof}

\begin{algorithm}[t]
        \caption{$\mathtt{AmplifiedAmplitudeEstimation}$: Estimation of probabilities with prior knowledge}
    \label{alg:class_ev}
    \KwIn{Reflection oracle $R_\Pi = \openone - 2 \Pi$, unitary state preparation oracle $O_\psi: \ket{0} \mapsto \ket{\psi}$, integer $\mu \geq 1$, allowable estimation error $\epsilon$, allowable failure probability $\delta'$.}
    \KwOut{An estimate of $\bra{\psi} \Pi \ket{\psi}$.}
    \begin{enumerate}
        \item Construct the operator $W' = -(\openone - 2 W^\mu O_\psi \ketbra{0}{0}O_\psi^\dag {W^\mu}^\dag)(\openone -2\Pi)$ \newline
        where $W = - R_\psi R_\Pi$ and $R_\psi = O_\psi \lb \One - 2 \ketbra{0}{0} \rb O_\psi^\dagger$ \;
        \item  Prepare the state $O_{\psi} \ket{0}$\;
        \item  Apply a unitary implementation of an amplitude estimation protocol, $U_{AE}(\epsilon, \delta')$, with $W'$ to $\ket{0}_a \ket{\psi}$ where $\epsilon, \delta'$ are the requested error and failure probabilities for the protocol and $\ket{0}_a$ is an ancilla register\;
        \item Let $\hat{P}_1$ denote the value of the phase returned in the ancilla register at the end of the amplitude estimation protocol\;
        \item $P_0 \gets \sin^2\left(\frac{\pi}{2(2\mu+1)}\right)$
        \item Classically compute and return $\sin^2\left(\frac{2\sin^{-1}(\sqrt{\hat{P}_1})\sin^{-1}(\sqrt{{P}_0})}{\pi}\right)$ as our estimate of $\bra{\psi} \Pi \ket{\psi}$.
    \end{enumerate}
\end{algorithm}

This means that to leverage the prior information, we perform an estimation procedure not directly on the expectation value but on the difference $\delta$ between the expectation value and our (known) upper bound on it. The advantage arises only in those cases where our upper bound $P_0$ is only slightly larger than the value of $\delta$. 
More specifically, we have an advantage over standard amplitude estimation if $P_0^{3/2} \in o(|\delta|)$. Note that above we assume that $\epsilon \in O(P_0^2/|\delta|)$. However, if we have the more stringent assumption that $\epsilon~\in~\Theta(P_0^2/|\delta|)$, then the scaling goes like
\begin{equation}
    N_{queries} \in O\left( \frac{\log(1/\delta')}{\sqrt{P_0}} \right),
\end{equation}
which demonstrates an advantage over what would be expected from Heisenberg limited scaling in this context. 
Put differently, if $|\delta|\in \Theta(P_0)$ and $P_0 \in \Theta(\epsilon)$, then we find $N_{\rm queries} \in O(1/\sqrt{\epsilon})$.

\begin{table}
    \centering
    \begin{tabular}{|l|l|}
        \hline
      \textbf{Symbol}   &  \textbf{Meaning}   \\ 
      \hline
        $\bra{\psi}\Pi \ket{\psi}$ & Expectation value (EV) we want to estimate \\
        $\mu \geq 1$ & Integer which quantifies our prior knowledge \\
        $P_0$ & Known upper bound on the target EV: $P_0 = \sin^2 (\pi/(2(2\mu +1)))$ \\
        $\delta<0$ & Defines how loose $P_0$ is: $\bra{\psi}\Pi \ket{\psi]}= P_0 +\delta$  \\
        $\epsilon':= P_1-\hat{P}_1 $ & Accuracy of the prior $\hat{P}_1$\\
        $\epsilon$ & Error in the estimation of $\delta$   \\
        $\delta'$ & Failure probability upper bound in estimation of $\delta$\\
        $\hat{\delta}$ &  $\epsilon$-precise estimate of $\delta$ \\
    \hline
    \end{tabular}
    \caption{Summary of variables - $\delta$ is the quantity we need to estimate.}
    \label{tab:my_label}
\end{table}

\subsection{Probability Estimation Using Fast-Square-Rootable Operators}

A challenge with amplified amplitude estimation is that we require the probability we wish to estimate to be small. Usually, probability estimation relies on techniques like the SWAP or Hadamard test. However, the Hadamard test does not return the expectation value directly but will return $(1/2 + {\rm Re}(\bra{\psi} U \ket{\psi})/2)$ for some unitary $U$ encoding the operator whose expectation value we wish to estimate. For small expectation values, $\bra{\psi} U \ket{\psi} \ll 1$, the amplitudes of the output state will be close to $1/2$ and our method from Lemma~\ref{lem:estimate_small_prob} cannot boost the result. 
Instead, we require a setup that directly encodes the small expectation values into an amplitude. Below, we provide a method for computing small expectation values by exploiting ideas from linear combinations of unitaries~\cite{childs2012hamiltonian}. Our method relies on the decomposition of $A$ into sums of matrices that are ``fast-square-rootable'', which means that there exists a constant query complexity algorithm for preparing the square root of the matrices given query access to the original matrix.

Let us now clarify our access model for the operators and quantum states whose expectation values we wish to estimate.
\begin{definition}[Access model]
\label{def:access}
    Define $O_{\psi}$ to be a unitary state preparation oracle such that $O_\psi \ket{0}^{\otimes n} = \ket{\psi}$ and the inverse, $O_{\psi}^\dagger \ket{\psi} = \ket{0}^{\otimes n}$, can be performed at the cost of a single query.
    Furthermore, let $A = \sum_{j=0}^{J-1} \alpha_j U_j$ be a unitary decomposition of a matrix $A \in \mathbb{C}^{2^n \times 2^n}$ such that $\alpha_j \geq 0$ and $U_j \in \mathbb{C}^{2^n \times 2^n}$ is unitary for all $j$.
    Define the block-encoding operations for $A$ as follows:
    \begin{equation}
        \mathtt{PREPARE}_A \ket{0}\ket{\psi} := \sum_{j=0}^{J-1} \frac{\sqrt{\alpha_j}}{\sqrt{\alpha}}\ket{j}\ket{\psi}, \qquad \mathtt{SELECT}_A := \sum_{j=0}^{J-1} \ketbra{j}{j}\otimes U_j,
    \end{equation}
    where $\alpha := \sum_{j=0}^{J-1} \alpha_j$. Then the unitary
    \begin{equation}
        U_A := \mathtt{PREPARE}_A^\dagger \, \mathtt{SELECT}_A \, \mathtt{PREPARE}_A
    \end{equation}
    block-encodes the matrix $A$ via $(\bra{0} \otimes \openone)U_A(\ket{0} \otimes \openone)= A/\alpha$.
\end{definition}

The following lemma shows how the expectation value of an operator can be encoded directly as a success probability.
\begin{lemma}[Expectation value as a success probability]
\label{lem:smalltrace}
Let $A_j$ be an operator such that $A_j = \sum_{k\ge 1} \beta_{jk} \Pi_{jk}$ where $\Pi_{jk}$ are projectors and $\beta_{jk} \geq 0$. Each projector is of the form $\Pi_{jk} = (\One - R_{jk})/2$ where $R_{jk}$ is a reflection operator. 
Let, for $y\in \{0,1\}$, 
\begin{equation}
    \mathtt{SELECT}_\pi\ket{k}\ket{y} \ket{\psi}\ket{\phi} = \ket{k} Z\ket{y} \lb R_{jk}^y\otimes (\ketbra{0}{k} + \ketbra{k}{0}) \rb \ket{\psi}\ket{\phi}
\end{equation}
and
\begin{equation}
    \mathtt{PREPARE}_{\pi}\ket{0}\ket{0} = \sum_{ky} \frac{({\beta_{jk}})^{1/4}}{\sqrt{2} \sqrt{\sum_{k} \sqrt{\beta_{jk}}}} \ket{k} \ket{y}
\end{equation}
be two unitary quantum oracles. Then there exists a unitary circuit, $U_\pi$, that can be implemented using a single query to each of $\mathtt{PREPARE}_\pi$, $\mathtt{PREPARE}_\pi^\dagger$ and $\mathtt{SELECT}_\pi$ such that for any state $\ket{\psi}$
$$
    {\rm Tr}(U_\pi(\ketbra{\psi}{\psi}\otimes \ketbra{0}{0}) U_\pi^\dagger (\One \otimes \ketbra{0}{0} )) = \frac{\bra{\psi} A_j \ket{\psi}}{\lb \sum_{k} \sqrt{\beta_{jk}} \rb^2}. 
$$
\end{lemma}
\begin{proof}
The result follows straightforwardly from the results of~\cite{childs2012hamiltonian} and~\cite{somma2013spectral} but we provide a short proof here for the interested reader.  Specifically, it uses intuition that the square root can be block encoded as
\begin{equation}
B_j := \sum_{{k\geq1}} \sqrt{\beta_{jk}} \Pi_{jk} \otimes (\ketbra{k}{0} + \ketbra{0}{k})
\end{equation}
since we then have that
\begin{equation}
B_j^2 = A_j \otimes \ketbra{0}{0} + \sum_{kk'\ge 1} {\sqrt{\beta_{jk}} \sqrt{\beta_{jk'}}} \Pi_{jk} \Pi_{jk'} \otimes \ketbra{k}{k'}.
\end{equation}
First, let us choose $U_\pi$ in the following manner:
\begin{equation}
    U_\pi := (\mathtt{PREPARE}_\pi\otimes \One)^\dagger \mathtt{SELECT}_\pi(\mathtt{PREPARE}_\pi\otimes \One).
\end{equation}
We then have that
\begin{align}
    U_\pi \ket{0}\ket{0}\ket{\psi} \ket{0} &=(\mathtt{PREPARE}_\pi\otimes \One)^\dagger \mathtt{SELECT}_\pi(\mathtt{PREPARE}_\pi\otimes \One)\ket{0}\ket{0}\ket{\psi} \ket{0}\nonumber\\
    &= (\mathtt{PREPARE}_\pi\otimes \One)^\dagger \mathtt{SELECT}_\pi \frac{\sum_{ky}(\beta_{jk})^{1/4}}{\sqrt{2} \sqrt{\sum_{k}\sqrt{\beta_{jk}}}} \ket{k}\ket{y} \ket{\psi} \ket{0}\nonumber\\
    &= (\mathtt{PREPARE}_\pi\otimes \One)^\dagger \frac{\sum_{ky}(\beta_{jk})^{1/4}}{\sqrt{2} \sqrt{\sum_{k}\sqrt{\beta_{jk}}}} \ket{k}Z\ket{y} R_{jk}^y\ket{\psi} \ket{k}\\\nonumber
    &= (\mathtt{PREPARE}_\pi\otimes \One)^\dagger \frac{\sum_{ky}(\beta_{jk})^{1/4}}{\sqrt{2} \sqrt{\sum_{k}\sqrt{\beta_{jk}}}} \ket{k}\ket{y} (-R_{jk})^y\ket{\psi} \ket{k}\nonumber.
\end{align}

Next, with this expression in place, we have that the unnormalized quantum state after having measured the $\mathtt{PREPARE}$ ancilla register as $\ket{00}$ is given by 
\begin{align}
    &\lb \bra{00} \otimes \One) (\mathtt{PREPARE}_\pi\otimes \One \rb^\dagger \frac{\sum_{ky}(\beta_{jk})^{1/4}}{\sqrt{2} \sqrt{\sum_{ky}\sqrt{\beta_{jk}}}} \ket{k}\ket{y} (-R_{jk})^y\ket{\psi} \ket{k}  \nonumber\\
    & =\sum_{k'y'ky} \frac{\beta_{jk}^{1/4} \beta_{jk'}^{1/4}}{2 \sum_{k} \sqrt{\beta_{jk}}} (\bra{k'} \bra{y'} \otimes \One) (\ket{k}\ket{y}(-R_{jk})^{y}) \ket{\psi} \ket{k}\nonumber\\
    & =\sum_{ky} \frac{\sqrt{\beta_{jk}} }{2 \sum_{k} \sqrt{\beta_{jk}}} (-R_{jk})^{y}) \ket{\psi} \ket{k}.
\end{align}
The probability of this occurring is by Born's rule the sum of the squares of the amplitudes which yields
\begin{equation}
    \norm{\sum_{ky} \frac{\sqrt{\beta_{jk}} }{2 \sum_{k} \sqrt{\beta_{jk}}} (-R_{jk})^{y} \ket{\psi} \ket{k}}^2= \sum_{k} \frac{\beta_{jk} {\bra{\psi} (\One - R_{jk})\ket{\psi}}}{2\lb \sum_k \sqrt{\beta_{jk}}\rb^2} =\frac{\bra{\psi} A_{j} \ket{\psi}}{\lb \sum_k \sqrt{\beta_{jk}}\rb^2}.
\end{equation}
\end{proof}

The above shows a straightforward experiment that can be performed to estimate the expectation value of any observable that is a convex combination of projection operators. This, unfortunately, is not general. While all operators can trivially be written as a sum of projectors, such a sum is not in general convex, and convexity is required to ensure that the signs of the combinations used in the above lemma combine to form the desired mean value. A natural tactic to deal with this is to express an arbitrary operator as a non-convex combination of convex sums of projectors. We use this tactic in the following result, where we provide an end-to-end algorithm for estimating the expectation value of such a non-convex combination, given that we know approximate expectation values for each of the convex sums.

\begin{theorem}[Application of AAE to fast-square-rootable operators]
\label{thm:AAE}
Let $A = \sum_{j=0}^{J-1} A_j = \sum_{jk} \beta_{jk} \Pi_{jk}$ such that for any $j'$, all $\beta_{j'k}$ are either non-negative or are all non-positive and $\Pi_{jk}$ are projectors. Assume we are provided a series of upper bounds $P_0(j)$ such that for all $j$
\begin{equation}
    P_0(j):=\sin^2 \lb \frac{\pi}{2(2\mu_j + 1)} \rb = \frac{\bra{\psi} A_j \ket{\psi}}{\lb \sum_k \sqrt{\beta_{jk}} \rb ^2} - \delta_j,
\end{equation}
with $\mu_j \geq 1$ being an integer and $\delta_j \in [-P_0(j), 0)$. Let $\|\beta\|_{1,1/2} := \sum_{j} \lb \sum_k \sqrt{\beta_{jk}} \rb^2$.
Then there exists a quantum algorithm that, for any $\ket{\psi}$ and $\epsilon \in O \lb \|\beta\|_{1,1/2} \min_j \{ P_0^2(j)/|\delta_j| \} \rb \subseteq o(1)$, can estimate $\bra{\psi} A \ket{\psi}$ within error $\epsilon$ and probability of failure at most $\delta'$ using a number of queries to $\mathtt{SELECT}_\pi$ and $\mathtt{PREPARE}_\pi$ that scales as
$$
    {O}\left( \frac{\sum_j\sqrt{P_0(j)} \|\beta\|_{1,1/2} \log(J/\delta')}{\epsilon } \max_j\left(\frac{P_0(j)}{|\delta_j|} \right) \right).
$$
\end{theorem}

\begin{proof}
The proof is an application of Lemma~\ref{lem:smalltrace} and Lemma~\ref{lem:estimate_small_prob}. The algorithm we propose for solving this problem is simple. We loop over each of the $A_j$ and measure each such operator within sufficient error to achieve our bounds. Specifically, if we wish to learn $\bra{\psi} A \ket{\psi}$ then it suffices to learn the individual estimates $\langle \hat{A}_j \rangle$ of the $\bra{\psi} A_j \ket{\psi}$ such that 
\begin{equation}
    |\bra{\psi} A \ket{\psi} - \sum_j \langle \hat{A}_j \rangle| \le \sum_j |\bra{\psi} A_j \ket{\psi} -  \langle \hat{A}_j \rangle| \le \epsilon.
\label{eq:epsbound}
\end{equation}
Let us now choose the error tolerance for our estimate $\langle \hat{A}_j \rangle$ such that 
\begin{equation}
    |\bra{\psi} A_j \ket{\psi} -  \langle \hat{A}_j \rangle| = \frac{ (\sum_k \sqrt{\beta_{jk}})^2\epsilon}{ \sum_{j}\left(\sum_k \sqrt{\beta_{jk}}\right)^2} =: \epsilon_j.
\label{eq:epschoice}
\end{equation}
It is clear that~\eqref{eq:epschoice} implies~\eqref{eq:epsbound}. Next, as we have inherited the assumptions of Lemma~\ref{lem:estimate_small_prob}, we can use that algorithm to learn the probability that the method of Lemma~\ref{lem:smalltrace} yields $\ket{00}$ in the ancilla register. Our expectation value estimate is found from that success probability by multiplying the success probability by $\lb \sum_{k} \sqrt{\beta_{jk}} \rb^2$. Thus, we must aim to estimate the success probability within error $\epsilon'_j$ such that
\begin{equation}
    \epsilon'_j \le \frac{\epsilon_j}{\lb \sum_{k} \sqrt{\beta_{jk}} \rb^2} = \frac{\epsilon}{\sum_{j}\left(\sum_k \sqrt{\beta_{jk}}\right)^2} =: \frac{\epsilon}{\|\beta\|_{1,1/2}}.
\label{interm_eps_bound}
\end{equation}
Then the cost of learning the probability within error $\epsilon'_j$ in terms of the number of queries made to our oracles scales from Lemma~\ref{lem:estimate_small_prob} and our assumption that $P_0^2(j)/|\delta_j| \subseteq o(1)$ as
\begin{equation}
    O\left(\frac{\sqrt{P_0(j)} \log(1/\delta'')}{\epsilon'_j } \left(\frac{P_0(j)}{|\delta_j|} \right)\right) = O\left( \frac{\sqrt{P_0(j)} \|\beta\|_{1,1/2} \log(1/\delta'')}{\epsilon } \left(\frac{P_0(j)}{|\delta_j|} \right)\right),
\end{equation}
where $\delta''$ is an upper bound on the failure probability and we chose $\epsilon'_j$ to saturate the upper bound in Eq.~\eqref{interm_eps_bound}.
This process needs to be repeated for each $A_j$ and so the total number of queries is the sum of this result over all $j$ which yields
$$
    {O}\left( \frac{\sum_j\sqrt{P_0(j)} \|\beta\|_{1,1/2} \log(J/\delta')}{\epsilon } \max_j\left(\frac{P_0(j)}{|\delta_j|} \right) \right).
$$
Note that the probability of error from the union bound is no longer bounded above by $\delta''$ but instead becomes $J\delta''$. In order to ensure that the overall failure probability is at most $\delta'$, it suffices to choose $\delta''=\delta'/J$ for each set which yields the above result. 
\end{proof}

This shows that in the event that the expectation values tend to zero and under the assumption that we have sufficient prior knowledge to at least know the leading digit of the result, the cost of learning the expectation value of an operator, which is given as a sum of projectors, can be substantially lower than the cost of performing na\"ive amplitude estimation which would scale as $\sum_{jk} \beta_{jk} \log(1/\delta') /\epsilon$. Under the assumption that we have constant relative uncertainty, $\max_j (P_0(j)/|\delta_j|) \in O(1)$, we then expect an asymptotic advantage if
\begin{equation}
\sum_j \sqrt{P_0(j)} \in o\left( \frac{\|\beta\|_{1,1}}{\|\beta\|_{1,1/2}}\right).
\end{equation}
Note that in general this condition will not be asymptotically attainable since $\|\beta\|_{1,1} \le \|\beta\|_{1,1/2}$ and so the right-hand side of the above expression can vanish as the dimension of the problem increases. This underscores the need to have applications where the probabilities that we wish to estimate are small. 

Our approach can be further generalized to compute a vector of expectation values straightforwardly. The additional cost involved is simply a factor of $m$ greater where $m$ is the number of distinct expectation values required. 
The cost of computing the vector of expectation values can be further reduced to a $O(\sqrt{m})$ scaling by using the methods of~\cite{huggins2021nearly}, which involves expressing the expectation value estimation as a gradient estimation problem which can be addressed from the results of~\cite{jordan2005fast,gilyen2019optimizing}. However, while such a process could be used to find a vector of independent observables over a set of points, it cannot be used to compute the values along a given path because of the way that prior knowledge is used in the amplified amplitude estimation algorithm. Specifically, the methods of~\cite{huggins2021nearly} compute each component of the vector of expectation values at the same time in superposition, meaning that this approach does not use any continuity of the values at adjacent points. Because of these complications, we focus in the following on the case of interpolating one-dimensional vectors while noting that similar results can be obtained for interpolating $m$-dimensional vectors.

The following corollary deals with the situation where we have prior knowledge of
the expectation value $\bra{\psi} A \ket{\psi}$ but the expectation value is too large for Lemma~\ref{lem:estimate_small_prob} or Theorem~\ref{thm:AAE} to be directly applicable. The main idea is to lower the expectation value of each of the operators $A_j$ by exploiting prior knowledge to reduce the number of projectors to be evaluated on a quantum computer.

\begin{corollary}[Using prior knowledge to refine estimates of large expectation values]
    Let $\epsilon > 0$ be an error tolerance and $A = \sum_{j=0}^{J-1} A_j = \sum_{j=0}^{J-1} \sum_{k=1}^{n_j} \beta_{jk} \Pi_{jk}$ an operator such that for any $j'$, all $\beta_{j'k}$ are either non-negative or are all non-positive and $\Pi_{jk}$ are projectors. Assume that for every $j$ there exists a subset $I_j \subseteq [n_j]$ of indices and known positive numbers $C_{j,i} \geq \widetilde{\epsilon}_j$, where 
    \begin{equation}
        \widetilde{\epsilon}_j := \frac{\epsilon}{2 \sum_{i \in I_j} \beta_{j,i}} \frac{(\sum_{k} \sqrt{\beta_{jk}})^2}{\sum_j (\sum_{k} \sqrt{\beta_{jk}})^2},
    \end{equation} 
    such that $|C_{j,i} - \bra{\psi} \Pi_{j,i} \ket{\psi}| \leq \widetilde{\epsilon}_j$.
    
    Furthermore, assume that we are provided with a series of upper bounds $P_0(j) = \sin^2(\pi/(2(2\mu_j+1)))$ and integers $\mu_j \geq 1$ such that
    \begin{equation}
        P_{A'_j} := \frac{\sum_{k \not\in I_j} \beta_{j,k} \bra{\psi} \Pi_{j,k} \ket{\psi}}{(\sum_{k \not\in I_j} \sqrt{\beta_{ji}})^2} = P_0(j) + \delta_j
    \end{equation}
    for $\delta_j \in [-P_0(j),0)$. Let $\|\beta\|_{1,1/2} := \sum_{j} \lb \sum_k \sqrt{\beta_{jk}} \rb^2$.
    Then there exists a quantum algorithm that, for any $\ket{\psi}$ and $\epsilon \in O \lb \|\beta\|_{1,1/2} \min_j \{ P_0^2(j)/|\delta_j| \} \rb \subseteq o(1)$, can estimate $\bra{\psi} A \ket{\psi}$ within error $\epsilon$ and probability of failure at most $\delta'$ using a number of queries to $\mathtt{SELECT}_\pi$ and $\mathtt{PREPARE}_\pi$ that scales as
    \begin{equation}
        {O}\left( \frac{\sum_j\sqrt{P_0(j)} \|\beta\|_{1,1/2} \log(J/\delta')}{\epsilon } \max_j\left(\frac{P_0(j)}{|\delta_j|} \right) \right).
    \end{equation}
\label{cor:large_val}
\end{corollary}

\begin{proof}
    The proof is similar to the proof of Theorem~\ref{thm:AAE}. As before, we wish to obtain an estimate $\hat{A}$ of $\bra{\psi}A\ket{\psi}$ such that
    \begin{equation}
        |\hat{A} - \bra{\psi}A\ket{\psi}| \leq \epsilon.
    \label{hatA_error}
    \end{equation}
    The assumption $ P_{A'_j} = \frac{\sum_{k \not\in I_j} \beta_{j,k} \bra{\psi} \Pi_{j,k} \ket{\psi}}{(\sum_{k \not\in I_j} \sqrt{\beta_{ji}})^2} = P_0(j) + \delta_j$ ensures that ignoring the set of projectors $\{\Pi_{j,i}\}_{i \in I_j}$ leads to a sufficiently low success probability $P_{A'_j}$ which can be estimated using Lemma~\ref{lem:estimate_small_prob}. From this we can obtain an $\epsilon'_j$-precise estimate $\hat{A}'_j$ of $A'_j := \sum_{k \not\in I_j} \beta_{j,k} \bra{\psi} \Pi_{j,k} \ket{\psi}$. This then allows us to obtain an estimate $\hat{A}_j$ for $A_j$ as follows:
    \begin{equation}
        \hat{A}_j = \hat{A}'_j + \sum_{i \in I_j} \beta_{j,i} C_{j,i}.
    \end{equation}
    The error associated with $\hat{A}_j$ is bounded by
    \begin{equation}
        |\hat{A}_j - \bra{\psi}A_j\ket{\psi}| \leq \epsilon'_j + \widetilde{\epsilon}_j \sum_{i \in I_j} \beta_{j,i}.
    \end{equation}
    We construct the estimate $\hat{A}$ via the estimates $\hat{A}_j$ of the individual terms $A_j$. Note that
    \begin{equation}
        |\hat{A} - \bra{\psi}A\ket{\psi}| = |\sum_j \hat{A}_j - \sum_j \bra{\psi}A_j\ket{\psi}| \leq \sum_j |\hat{A}_j - \bra{\psi}A_j\ket{\psi}|.
    \end{equation}
    Ensuring that
    \begin{equation}
        |\hat{A}_j - \bra{\psi}A_j\ket{\psi}| \leq \frac{(\sum_{k} \sqrt{\beta_{jk}})^2}{\sum_j (\sum_{k} \sqrt{\beta_{jk}})^2} \epsilon =: \epsilon_j
    \end{equation}
    automatically yields \eqref{hatA_error}. This implies that it suffices to choose $\epsilon'_j = \epsilon_j/2$ and 
    \begin{equation}
        \widetilde{\epsilon}_j = \frac{\epsilon_j}{2 \sum_{i \in I_j} \beta_{j,i}}.
    \end{equation}
    Note that $\widetilde{\epsilon}_j$ only affects the cost of obtaining our prior knowledge but not the quantum cost associated with estimating the operators $A_j$. The quantum experiments yield estimates of the probabilities $P_{A'_j}$ rather than the estimates for $\hat{A}'_j$ directly. It suffices to estimate $P_{A'_j}$ within error 
    \begin{equation}
        \epsilon_{P_j} := \frac{\epsilon'_j}{\lb \sum_{k \not\in I_j} \sqrt{\beta_{jk}} \rb^2} = \frac{\epsilon}{2 \sum_j (\sum_{k} \sqrt{\beta_{jk}})^2} \frac{(\sum_{k} \sqrt{\beta_{jk}})^2}{\lb \sum_{k \not\in I_j} \sqrt{\beta_{jk}} \rb^2} \geq \frac{\epsilon}{2\|\beta\|_{1,1/2}}
    \end{equation}
    to obtain an $\epsilon'_j$-precise estimate of $\hat{A}'_j$.
    Thus, following the same argument as in the proof of Theorem~\ref{thm:AAE}, we find that the number of queries to the unitaries used for implementing the projectors, $\mathtt{PREPARE}_{\Pi}$ and $\mathtt{SELECT}_{\Pi}$, is in
    \begin{equation}
        O \left( \frac{\sum_j\sqrt{P_0(j)} \|\beta\|_{1,1/2} \log(J/\delta')}{\epsilon } \max_j\left(\frac{P_0(j)}{|\delta_j|} \right) \right),
    \end{equation}
    where $\|\beta\|_{1,1/2}=\sum_{j} (\sum_k \sqrt{\beta_{jk}})^2$.
\end{proof}

The above corollary shows that we can get an asymptotic advantage over standard amplitude estimation even for large expectation values as long as we have sufficient prior knowledge of the individual expectation values in the projector decomposition. This allows us to express a large expectation value that we wish to compute as a sequence of smaller ones, each of which we then estimate using amplified amplitude estimation to achieve an advantage relative to what would be possible with methods that do not use any prior information.

\section{Applications}
\label{sec:applications}

There are many different ways in which prior knowledge about the expectation value of an observable on a wave function can be obtained. In what follows, we will explore two examples of how prior knowledge can be exploited by amplified amplitude estimation to solve problems relevant to chemistry. However, the applications of AAE are potentially broader than the simulation of quantum systems as long as the conditions on the expectation values and the prior knowledge are satisfied. 

The first example we will consider involves using classical methods to obtain a low-cost a priori estimate of the expectation value of a given observable and apply AAE to improve on such prior estimate. This is a common situation in simulations of molecular systems: often, a cheap mean-field solution such as Hartree-Fock (HF) or Density Functional Theory (DFT) forms the basis for a more demanding correlated calculation. The expectation values from these low-cost methods, together with a tight enough bound on their error, are the two ingredients needed to apply our algorithm. Specifically, we will show how to use AAE to reduce the cost of estimating expectation values of one-body fermionic operators, such as dipole moments for water clusters.

The second example focuses on estimating properties of a target quantum system by exploiting available prior knowledge regarding other quantum systems with similar parameters.
For instance, in Markov Chain Monte Carlo dynamics, new molecular geometries are proposed by randomly changing the current coordinates. These new geometries cannot differ too much in energy from the current system, or the detailed balance mechanism would reject them. The prior knowledge in this setting can be exploited to infer the energy for the new geometry by using knowledge about the current energy. We can lower the cost even further by realizing that derivative operators may have a lower one-norm than the Hamiltonian operator. This motivates us to consider the problem of computing the ground state energy of a molecular system within the Born-Oppenheimer approximation by integrating the directional derivative of the energy using Newton-Cotes formulas.

\subsection{Fermionic operator estimation with prior knowledge from low-cost methods}
\label{sec:dipole}

To apply our prior results to the estimation of expectation values of fermionic operators, we first need to reduce the problem from one of measuring fermionic operators to a problem involving measuring projections. We will demonstrate this for a general fermionic one-body operator
\begin{equation}
    A = \sum_{pq} A_{pq} a^{\dagger}_p a_q,
\end{equation}
where $p$ and $q$ label a set of spin orbitals such that $a^{\dagger}_p$ is the fermionic creation operator acting on orbital $p$ and $a_q$ is the fermionic annihilation operator acting on orbital $q$.
After using the Jordan-Wigner decomposition, we have that 
\begin{equation}
    A = \sum_{p<q} \frac{1}{2}A_{pq} X_p \otimes Z_{p+1} \cdots Z_{q-1}\otimes X_q + \frac{1}{2}A_{pq} Y_p \otimes Z_{p+1} \cdots Z_{q-1}\otimes Y_q + \sum_{p} A_{pp}\frac{1-Z_p}{2}.
\end{equation}
The Pauli products in this formula are reflection operators, each of which can be thought of as a projection onto a negative eigenspace added to a projection onto a positive eigenspace. Let us define, for example,
\begin{align}
    \Pi_{p,q}^{(\pm X)} &:= \frac{\openone_{pq} \pm X_p \otimes Z_{p+1} \cdots Z_{q-1}\otimes X_q}{2}, \\
    \Pi_{p,q}^{(\pm Y)} &:= \frac{\openone_{pq} \pm Y_p \otimes Z_{p+1} \cdots Z_{q-1}\otimes Y_q}{2}, \\
    \Pi_{p,p}^{(\pm Z)} &:= \frac{\openone_{pp} \pm Z_p}{2},
\end{align}
where the $\pm$ denotes the projector onto the positive and negative eigenspaces of the operator. This is an obvious choice that does not necessarily have the best properties for the AAE algorithm. Optimizing this choice might lead to better scaling of the algorithm. Using the fact that $\Pi^{(+)} + \Pi^{(-)} = \mathbb{1}$ we then have that
\begin{align}
    \bra{\psi} A \ket{\psi} &= \frac{1}{2}\sum_{p < q} A_{pq} \bra{\psi}\left( \Pi_{p,q}^{(+X)} - \Pi_{p,q}^{(-X)} + \Pi_{p,q}^{(+Y)} - \Pi_{p,q}^{(-Y)}\right)\ket{\psi}  + \sum_p \bra{\psi}A_{pp}\Pi_{p}^{(-Z)}\ket{\psi} \nonumber\\
    &=\sum_{p < q} A_{pq} \bra{\psi}\left( \Pi_{p,q}^{(+X)}  + \Pi_{p,q}^{(+Y)} \right)\ket{\psi}  + \sum_p A_{pp}\bra{\psi}\Pi_{p}^{(-Z)}\ket{\psi} - \sum_{p<q} A_{pq}.
\end{align}
Even though the operator is now written as a sum of projectors, the sum is still not necessarily convex. If terms are negative, then we can always group them separately from the positive terms and apply two separate rounds of AAE. The above expectation value can then be decomposed as follows:
\begin{align}
    \bra{\psi} A \ket{\psi} &= \sum_{j=0,1}\sum_{p < q} \frac{A_{pq} + (-1)^j\left\vert A_{pq}\right\vert}{2} \bra{\psi}\left( \Pi_{p,q}^{(+X)}  + \Pi_{p,q}^{(+Y)} \right)\ket{\psi} \\
    & \quad + \sum_{j=0,1} \sum_p \frac{A_{pp}+(-1)^j|A_{pp}|}{2}\bra{\psi}\Pi_{p}^{(-Z)}\ket{\psi} - \sum_{p<q} A_{pq}\nonumber\\
    &=: \sum_{j}\sum_{k}\beta_{jk}\Pi_{jk} - \sum_{p<q} A_{pq},\label{eq:Proj_decomp}
\end{align}
where the last step is a relabeling of the two separate convex sums into a single sum over $k$. Further, the constant offset $\sum_{p<q} A_{pq}$ is known a priori and so does not need to be considered for our purposes.  This shows that we can use Theorem \ref{thm:AAE} to bound the query complexity of estimating $\bra{\psi} A \ket{\psi}$.

Let us now consider the example of estimating the expectation value of one-body operators for groups of water molecules of increasing size. We will consider three types of operators: the $x,y,z$ components of the dipole operator, the kinetic energy operator, and the one-body part of the Hamiltonian. The matrix elements of the components of the dipole operator are obtained as integrals over the molecular orbitals $\phi_p$ by
\begin{align}
D_{\rho,pq} = \int \phi^*_p\left(\mathbf{r}\right)\,\rho\, \phi_q\left(\mathbf{r}\right)\,\mathrm{d}\mathbf{r} \mathrm{\quad for} \; \rho \in \{x,y,z\}.
\end{align}
The value of the dipole moment is often significant because it 
provides qualitative information about the strength of intermolecular forces in the system, and the expectation value of the dipole operator in the ground state will give an estimate of this dipole moment.
Similarly, the kinetic energy operator $K$, with elements
\begin{align}
K_{pq} = \int \phi^*_p\left(\mathbf{r}\right)\,\left(-\frac{1}{2}\nabla^2\right)\, \phi_q\left(\mathbf{r}\right)\,\mathrm{d}\mathbf{r},
\end{align}
and the one-body Hamiltonian $P$, with elements
\begin{align}
P_{pq} = K_{pq} + \int \phi^*_p\left(\mathbf{r}\right)\,V\left(\mathbf{r}\right)\, \phi_q\left(\mathbf{r}\right)\,\mathrm{d}\mathbf{r},
\end{align}
where $V\left(\mathbf{r}\right)$ is the external potential due to attraction with the nuclei, are important operators that help assess the quality of a wave function.

In order to use amplified amplitude estimation to profitably estimate a quantity, it is important that the requirements on the size of the perturbation ($\delta_j \in [-P_0(j),0)$) and the prior knowledge ($P_0(j)$) s.t. $\epsilon~\in~O(P_0^2(j)/|\delta_j|)$) imposed in Corollary~\ref{cor:large_val} are met. Here we study these requirements numerically. In Figure~\ref{fig:delta}, we explore the spread of $\delta_j$ for the problem of estimating the dipole moment, the kinetic energy, and the single-particle Hamiltonian for a cluster of water molecules. Specifically, we plot the spread of the perturbation as a function of a growing water cluster for prior knowledge obtained either from RHF or DFT with the B3LYP functional or low bond dimension DMRG. We use the projector decomposition of the single-particle operators as described in Eq.~\eqref{eq:Proj_decomp}. 
It is clear that different methods (depending on their accuracy) will give different orders of magnitude for the $\delta_j$: while RHF and DFT require learning a deviation of the same order of magnitude, low bond dimension DMRG is more accurate, and the corresponding correction of the prior knowledge is two orders of magnitude smaller. In all these cases, the value of $\delta_j$ that is attainable is sufficiently small so that we can profitably use this prior knowledge provided that $P_0(j)/|\delta_j|$ is appropriately small. In particular, the difference between the DMRG calculations and the true values for the expectation values of the projectors differ on the order of $10^{-9}$ for all data between $2$ and $7$ water molecules in our box, whereas RHF and DFT calculations provide estimates that are accurate within $10^{-7}$ and $10^{-6}$ roughly. For RHF and DFT, we see evidence of polynomial decay in the difference as the number of water molecules increases. The data is roughly consistent with a polynomial scaling of $O((\# \text{ waters})^{-2})$ for RHF and for DFT.  In contrast, we see no evidence of a clear systematic trend for the data for the DMRG calculation results at the scale of our computations.

\begin{figure}[!ht]
\includegraphics[width=0.48\linewidth]{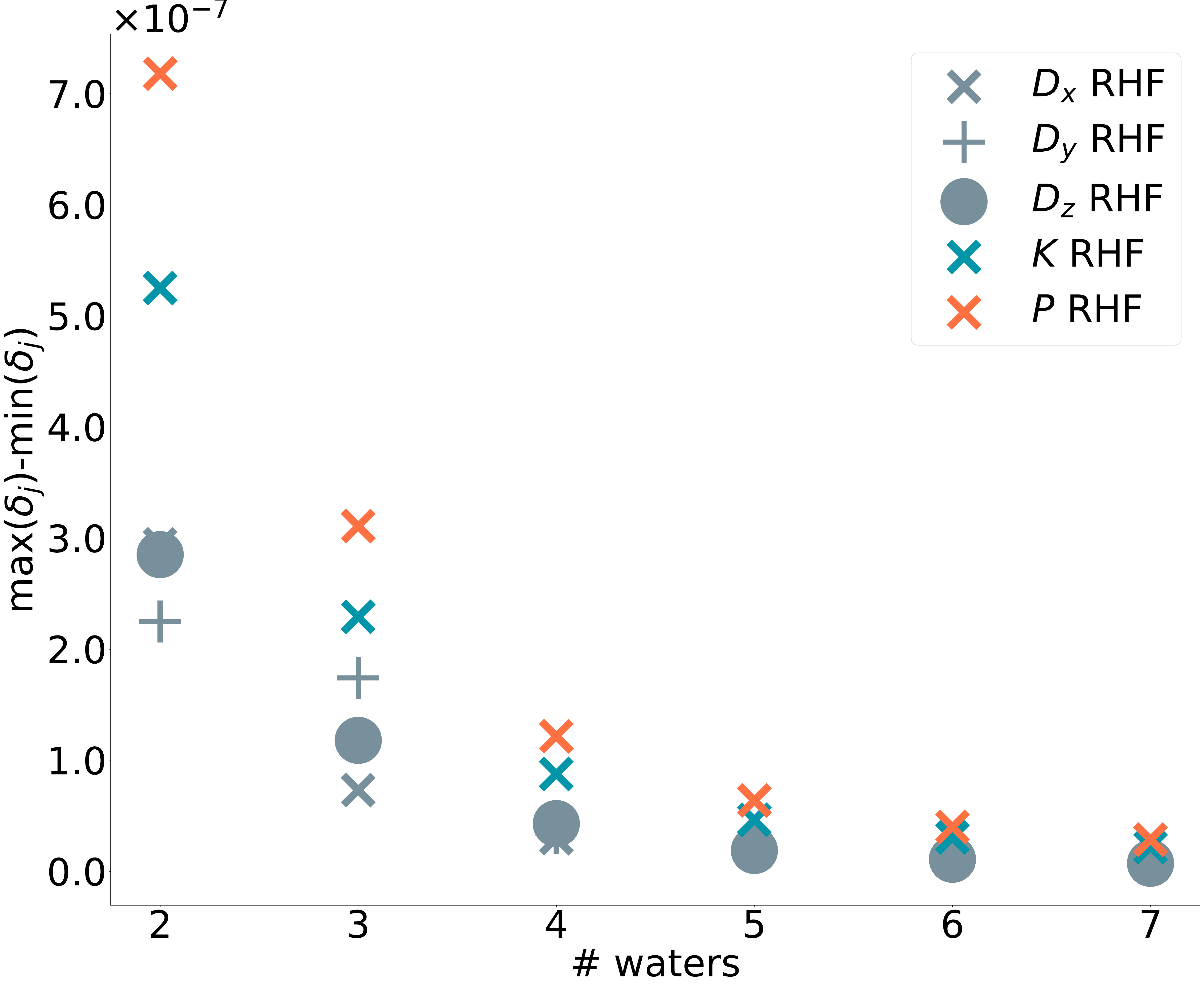}
\includegraphics[width=0.48\linewidth]{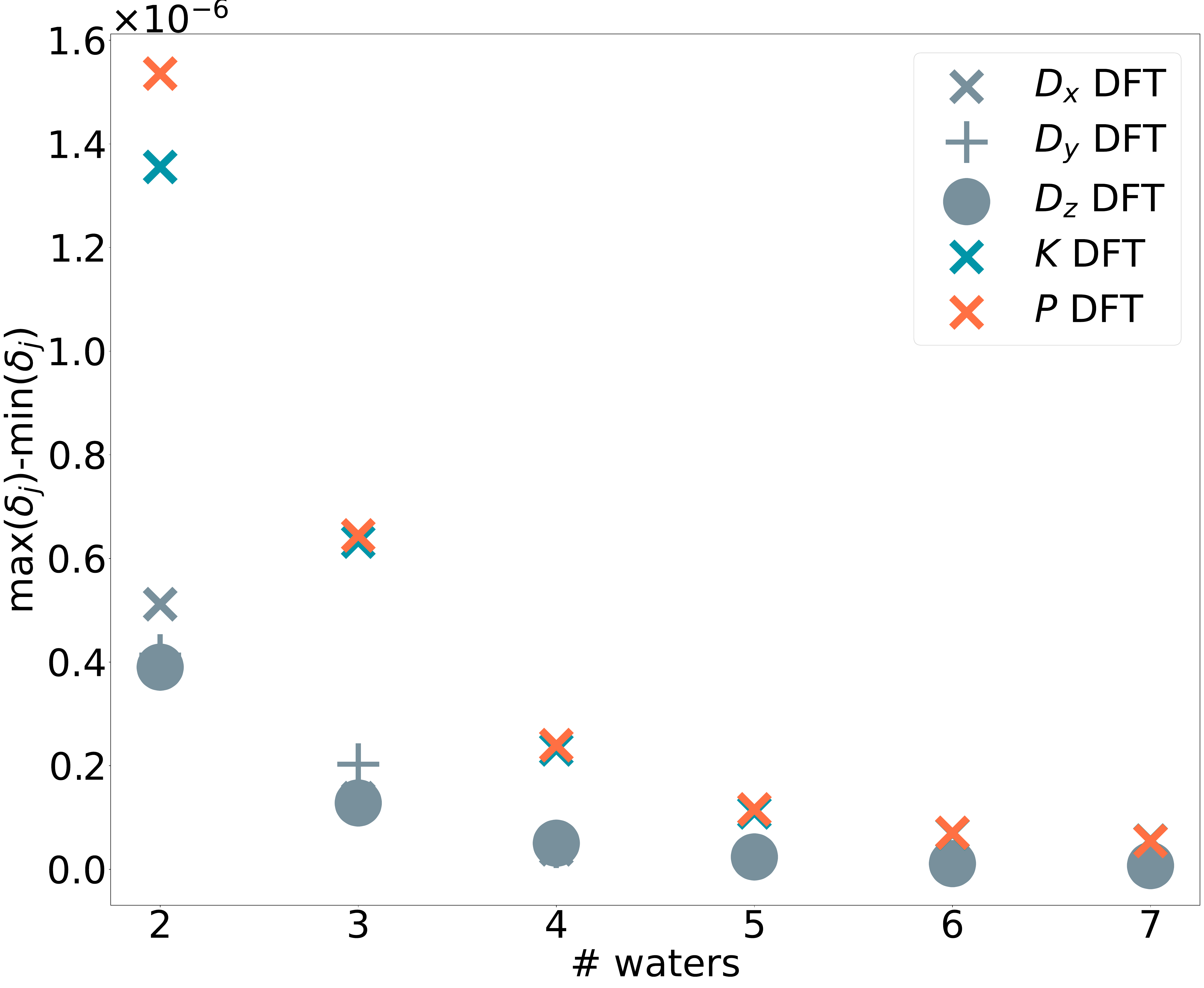}\\
\includegraphics[width=0.48\linewidth]{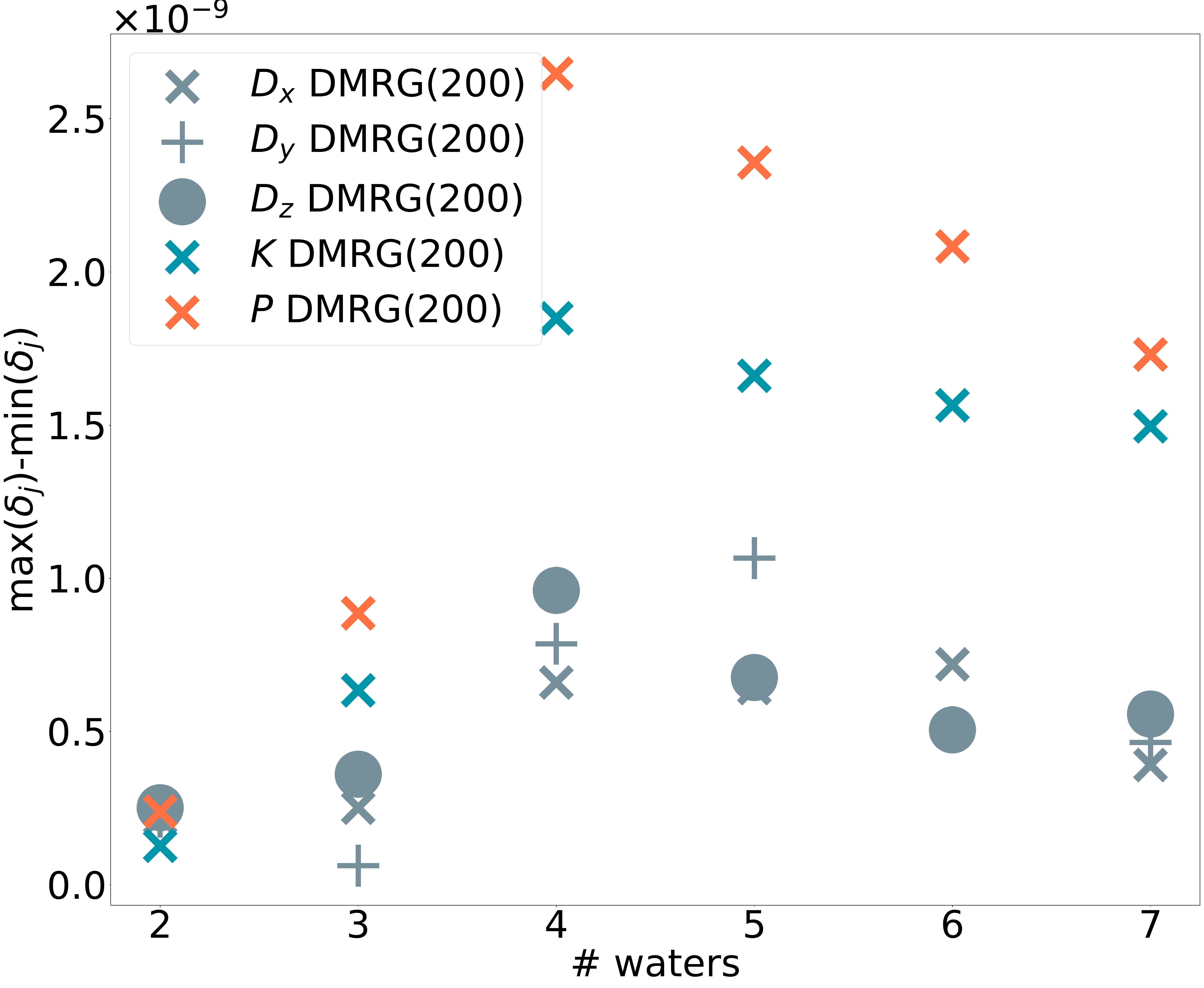}
\caption{Each plot shows the spread of $\delta_j$ for five different single-particle operators versus the size of a growing cluster of water molecules when using prior information. The prior information is obtained from restricted Hartree-Fock (RHF), density functional theory (DFT) with the B3LYP functional, and density matrix renormalization group theory (DMRG) with a low bond dimension of 200, respectively. The ground truth to which this prior information is compared is obtained from DMRG with a high bond dimension of 750. The molecular properties are calculated in the cc-pVDZ basis using an all-electron approach for RHF and DFT and an active space correlating 6 orbitals and 8 electrons per water molecule for the DMRG calculations. The operators are the three components of the dipole moment ($D_x$, $D_y$, and $D_z$), the kinetic energy operator ($K$) and the single-particle Hamiltonian ($P$). The minimum and maximum values of $\delta_j$ are taken over the two projectors in the single-particle operator ($j=0,1$) and over an ensemble of 26 geometries taken from a well-equilibrated $500$ step molecular dynamics (MD) run at $300 \mathrm{K}$ taking $1\mathrm{fs}$ steps using the TIP3P water model. The decrease in value of the $\delta_j$ with growing water cluster can be attributed to the normalization of the projectors in a growing basis. The different behavior for the low bond dimension DMRG at small cluster sizes is due to the equivalence between low bond dimension DMRG and high bond dimension DMRG for these systems. The prior knowledge in that case is essentially exact and the size of the perturbation to be learned is close to $0$. The three components of the dipole operator should yield similar results since the water clusters should be near-rotationally invariant over the course of the MD run. Finite size effects in the smaller clusters and the finite ensemble of geometries result in some spread between the three components. MD calculations where performed with ASE~\cite{Hjorth2017ASE}, HF and DFT calculations with PySCF~\cite{Sun2020PySCF} and DMRG calculations with Block2~\cite{Zhai2021Low,Zhai2023Block2}.}
\label{fig:delta}
\end{figure}

The relative size of $\delta_j$ and $P_0(j)$ is a primary driver of the complexity of our method.
Figure~\ref{fig:delta_by_p} addresses the question of relative error by examining the difference between the relative errors estimated using RHF, DFT and DMRG.  Interestingly, the data shows a different scaling behavior between these results: For the first two methods, the ratio goes down as the water clusters grow, while for the last it goes up (at least up to the sizes we were able to simulate). For the DMRG, this means that the prior knowledge becomes of decreasing quality and the correction becomes more important as the water cluster grows. This is to be expected as the bond dimension in the DMRG calculations is not growing with system size.

We find in this data that the relative error for these predictions is on the order of $10^{-3}$ for the case of restricted Hartree-Fock and DFT; whereas DMRG performs again substantially better giving errors on the order of $10^{-5}$.  In these cases, the error is sufficiently small that we can reliably pick the $\delta_j$ small enough to ensure that $\max(P_0(j)/\delta_j)\approx 1$. We further see that in the case of DMRG the error seems to be increasing as we approach the thermodynamic limit.  RHF seems to be decreasing at a rate that is approximately consistent with an $O((\# \text{waters})^{-1/4})$ scaling, but we see evidence of saturation in the DFT data. While the trends in this data are difficult to divine, this provides evidence that advantages can be seen as long as $\sum_{j} \sqrt{P_0(j)} \in o(1)$.

\begin{figure}[!ht]
\includegraphics[width=0.48\linewidth]{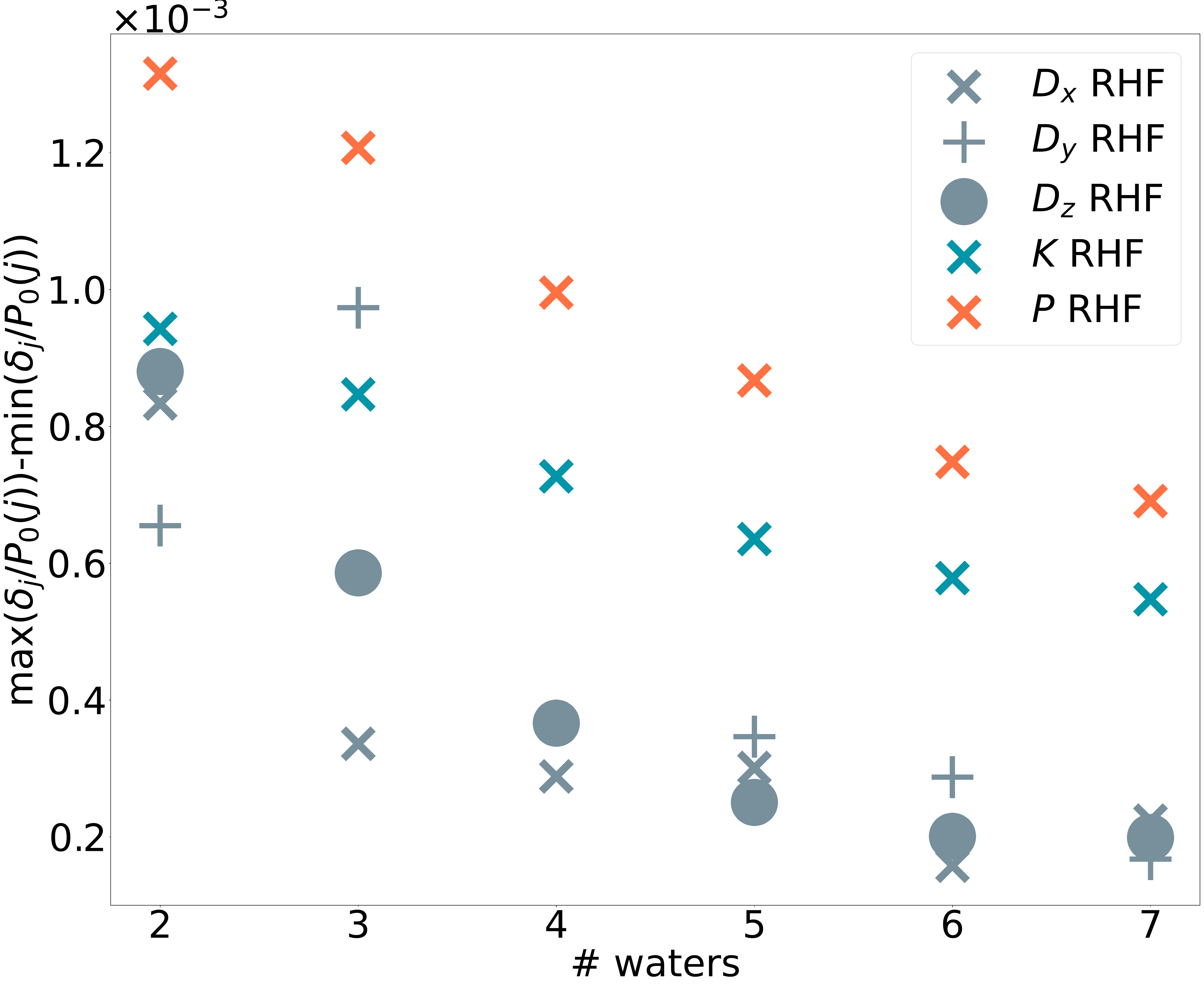}
\includegraphics[width=0.48\linewidth]{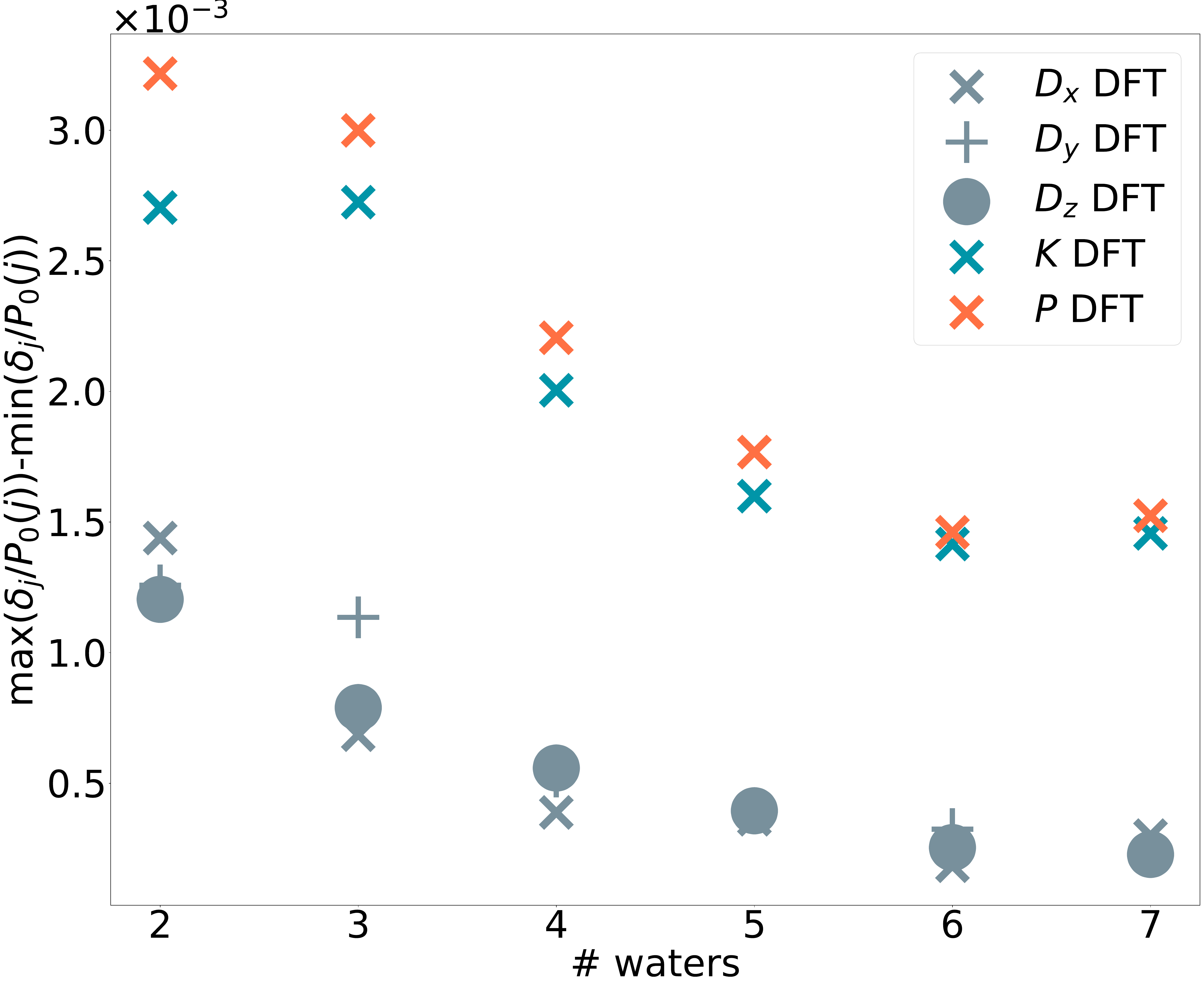} \\
\includegraphics[width=0.48\linewidth]{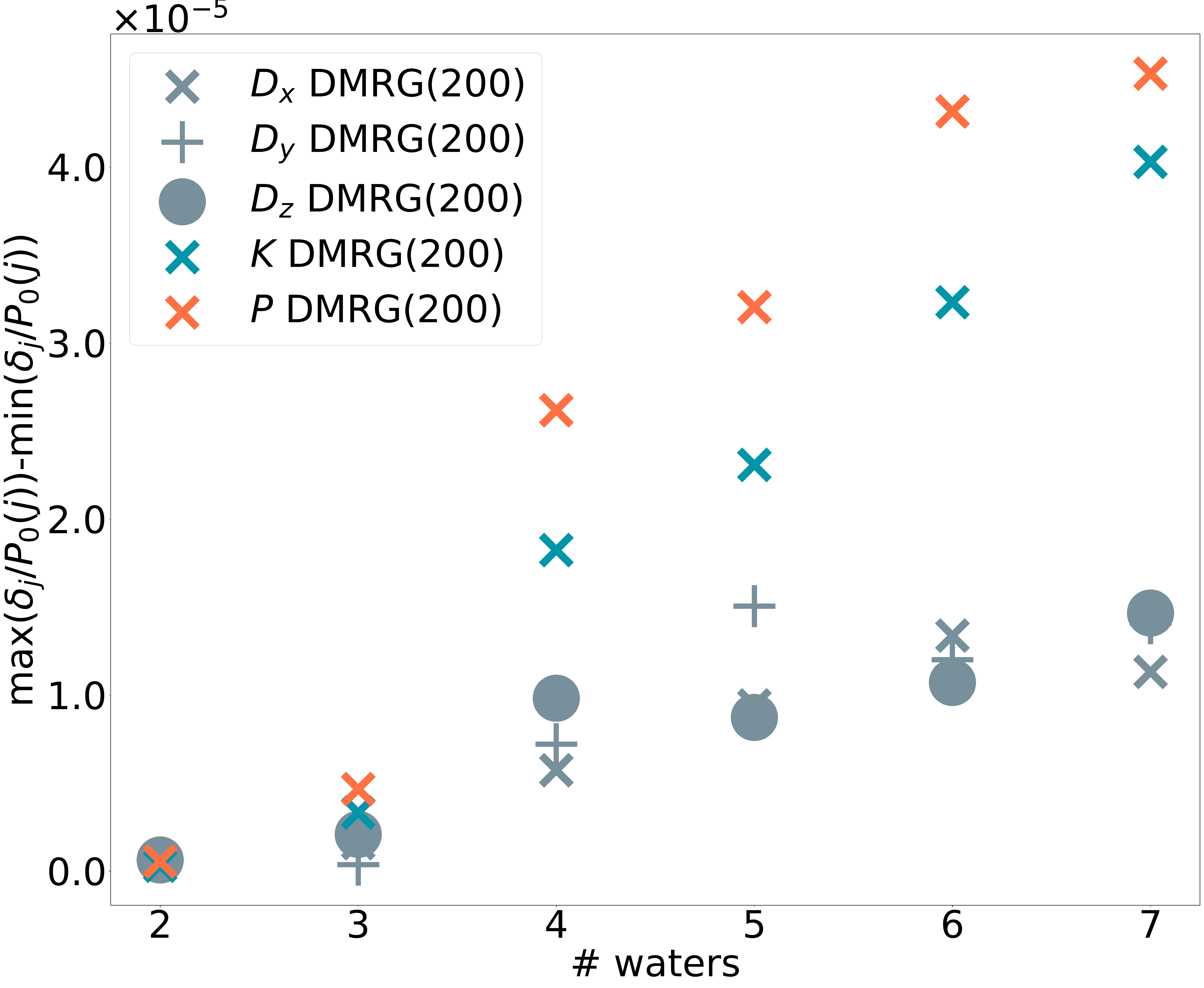}
\caption{Each plot shows the spread of $\delta_j/P_0(j)$ for five different single-particle operators versus the size of a growing cluster of water molecules when using prior information. The prior information is obtained from restricted Hartree-Fock (RHF), density functional theory (DFT) with the B3LYP functional, and density matrix renormalization group theory (DMRG) with a low bond dimension of 200, respectively. The ground truth to which this prior information is compared is obtained from DMRG with a high bond dimension of 750. The molecular properties are calculated in the cc-pVDZ basis using an all-electron approach for RHF and DFT and an active space correlating $6$ orbitals and 8 electrons per water molecule for the DMRG calculations. The operators are the three components of the dipole moment ($D_x$, $D_y$, and $D_z$), the kinetic energy operator ($K$) and the single-particle Hamiltonian ($P$). The minimum and maximum values of $\delta_j$ are taken over the two projectors in the single-particle operator ($j=0,1$) and over an ensemble of 26 geometries taken from a well-equilibrated $500$ step molecular dynamics (MD) run at $300 \mathrm{K}$ taking $1\mathrm{fs}$ steps using the TIP3P~\cite{Jorgensen1983TIP3P} water model.}
\label{fig:delta_by_p}
\end{figure}

One caveat in the above discussion is that we assume for generality that the quantum state $\ket{\psi}$ is provided by an oracle. However, in practice, particular methods must be considered to prepare such a state. In the following, we will consider the case where $\ket{\psi}$ corresponds to the ground state of the system of interest. 
While many methods can be considered for preparing an approximate ground state, adiabatic state preparation and phase estimation are two of the most popular methods for achieving a small preparation error. In both cases, the resources needed scale polynomially with the inverse eigenvalue gap. In the case of adiabatic state preparation, this is because of the dependence on the inverse gap. For phase estimation, the resource scaling arises from the cost of measuring the energy with an error smaller than the gap. Here, we take the latter approach as it will often provide superior scaling with the eigenvalue gap, only tying the performance in cases where boundary cancellation methods are employed and we are in a regime where the evolution time is sufficiently long. Specifically, we use the results of~\cite{Ge2019groundstate}, which we restate as the following lemma.

\begin{lemma}[Ground state preparation, Ge et al.~\cite{Ge2019groundstate}]
\label{lem:Ge}
For any $\epsilon_{\psi} > 0$, assume we have access to a trivially preparable state $\ket{\psi_0'}\in \mathbb{C}^{2^n}$ that has overlap at least $|\braket{\psi_0'}{\psi_0}| \geq |a_0| > 0$ with the target state $\ket{\psi_0}$, which is an eigenstate of a Hamiltonian $H = \sum_{k} \alpha_{k} U_k$ for unitary $U_k\in \mathbb{C}^{2^n\times 2^n}$ with coefficient $1$-norm $\|\alpha\|_1$. Furthermore, assume we have an eigenvalue estimate $E_0'$ of the eigenvalue $E_0$ associated with $\ket{\psi_0}$. We then further assume that $|E_0 - E_0'| \in O(\gamma/\log(|a_0| \epsilon_{\psi}))$ where $\gamma$ is a lower bound on the spectral gap of $H$. Then there exists a unitary $U_{\psi_0}$ acting on $b+n$ qubits such that 
\begin{enumerate}
\item We can perform a heralded measurement such that upon measuring zero on a qubit register, 
$$
\left|\frac{(\bra{0}^b \otimes \One) U_{\psi_0}\ket{0}^b\ket{\psi_0'}}{\norm{(\bra{0}^b \otimes \One) U_{\psi_0}\ket{0}^b\ket{\psi_0'}}} - \ket{\psi_0}\right|_2 \le \epsilon_\psi.
$$
\item The probability of not measuring zero is at most $1-\epsilon_{\psi}^2$.
\item $U_{\psi_0}$ can be implemented using a number of queries to the state preparation routine for $\ket{\psi'_0}$ that scales as
$$
    \widetilde{O}\left( \frac{\log(1/\epsilon_\psi)}{|a_0|} \right)
$$
and a number of queries to $\mathtt{PREPARE}_H$ and $\mathtt{SELECT}_H$ that scale as 
$$
\widetilde{O}\left(\frac{\|\alpha\|_1}{|a_0|\gamma} \right).
$$
\end{enumerate}

\end{lemma}

For our problem, let us assume that we are provided an input ground state $\ket{\psi}$ and wish to estimate a one-body operator $\bra{\psi} A \ket{\psi}$ given an \emph{a priori} estimate of the operator that is provided by a classic approach such as density functional theory, Hartree-Fock theory or a low bond dimension tensor network approximation.  Our aim is then to refine this knowledge into an improved estimate of the quantity using amplified amplitude estimation.

The following lemma shows how we can, under certain conditions, estimate $\bra{\psi} A \ket{\psi}$ using $O \lb 1/ \sqrt{\epsilon}\rb$ queries to the underlying $\mathtt{PREPARE}$ and $\mathtt{SELECT}$ unitaries.

\begin{lemma}[Beating the Heisenberg limit with prior knowledge]
\label{lem:D}
    Let $\epsilon > 0$ be an error tolerance and $H=\sum_j \alpha_j U_j$, where $U_j \in \mathbb{C}^{2^n \times 2^n}$ is unitary, be a Hamiltonian with ground state $\ket{\psi}$ and coefficient 1-norm $\norm{\alpha}_1$. Let $A~=~\sum_{j=1}^{N_\pi}~\beta_j~\Pi_j$ be a projector decomposition of the operator $A$ with $\Pi_j$ being a projection operator that is implemented via a reflection operator $R_j$. Assume we are provided with a subset $I \subseteq [N_\pi]$ of indices and non-negative numbers $C_{i} \geq \widetilde{\epsilon}$ where $i \in I$ and $\widetilde{\epsilon} := \frac{\epsilon}{2\sum_j |\beta_j|}$ such that $|C_{i} - \bra{\psi} \Pi_{i} \ket{\psi}| \leq \widetilde{\epsilon}$. Furthermore, assume that we have a series of upper bounds $P_0(k) = \sin^2(\pi/(2(2\mu_k+1)))$ for $k \in [N_\pi] \setminus I$ and integers $\mu_k \geq 1$ such that
    \begin{equation}
        P_{k} := \bra{\psi} \Pi_{k} \ket{\psi} = P_0(k) + \delta_k
    \end{equation}
    where $\delta \in [-P_0(k),0)$. Let $\norm{\beta}_1 := \sum_{j=1}^{N_\pi} |\beta_j|$.
    Then under the assumptions of Lemma~\ref{lem:Ge} and if $\epsilon \in O \lb \norm{\beta}_{1} \min_k \{ P_0^2(k)/|\delta_k| \} \rb \subseteq o(1)$ and $P_0(k) \in O \lb \epsilon \rb$ for all $k~\in~[N_\pi]~\setminus~I$, we can learn $\bra{\psi} A \ket{\psi}$ within error $\epsilon$ and probability of failure at most $\delta'$ using a number of queries to $\mathtt{SELECT}_H$, $\mathtt{SELECT}_\Pi$, $\mathtt{PREPARE}_H$ and $\mathtt{PREPARE}_\Pi$ that scales as
    \begin{equation}
        \widetilde{O}\left( \frac{  N_\pi' \|\alpha\|_1 \norm{\beta}_{1} \log(N_\pi'/\delta')}{\gamma|a_0|\sqrt{\epsilon} } \max_{k \not\in I}\left(\frac{P_0(k)}{|\delta_k|} \right) \right),
    \end{equation}
    where $N_\pi' := |[N_\pi] \setminus I|$.
\end{lemma}

\begin{proof}
    Using Corollary~\ref{cor:large_val} where now each $A_j$ involves only a single projector $\Pi_j$, we find that the number of queries to $\mathtt{PREPARE}_\Pi$ and $\mathtt{SELECT}_\Pi$ is in
    \begin{equation}
        {O}\left( \frac{\sum_{k \not\in I}\sqrt{P_0(k)} \|\beta\|_{1} \log(N_\pi'/\delta')}{\epsilon} \max_{k \not\in I}\left(\frac{P_0(k)}{|\delta_k|} \right) \right).
    \end{equation}

    Next, we need to consider the cost of ground state preparation for the algorithm. Each state preparation can be performed using quantum phase estimation with $H=\sum_j \alpha_j U_j$, which is block-encoded using $\mathtt{SELECT}_H$ and $\mathtt{PREPARE}_H$. We then use this process and a comparator to determine whether the energy is less than or equal to the threshold needed to determine whether the state is indeed the ground state. 
    Then from Lemma~\ref{lem:Ge} the cost per query to the state preparation algorithm is $\tilde{O}(\|\alpha\|_1/|a_0|\gamma)$. Thus, the overall number of queries is
    \begin{equation}
    \begin{split}
        \widetilde{O}\left( \frac{\sum_{k \not\in I}\sqrt{P_0(k)} \norm{\alpha}_1 \|\beta\|_{1} \log(N_\pi'/\delta')}{\gamma |a_0| \epsilon} \max_{k \not\in I}\left(\frac{P_0(k)}{|\delta_k|} \right) \right) \\ 
        \subseteq \widetilde{O}\left( \frac{N_\pi' \norm{\alpha}_1 \|\beta\|_{1} \log(N_\pi'/\delta')}{\gamma |a_0| \sqrt{\epsilon}} \max_{k \not\in I}\left(\frac{P_0(k)}{|\delta_k|} \right) \right),
    \end{split}
    \end{equation}
    where we used the assumption that $P_0(k) \in O \lb \epsilon \rb$.
\end{proof}
The main idea used in the above lemma is to split the projector decomposition of $A$ into two disjoint sets. The first set consists of projectors $\Pi_i$ for which we have $\widetilde{\epsilon}$-precise estimates $C_i$ of their expectation values. Importantly, these estimates are sufficiently precise so that we do not need a quantum computer to refine them. The other set contains the projectors whose expectation values are small enough such that their estimates can be refined via AAE.

The gate complexity of the above oracle queries depends strongly on the basis set chosen to represent the problem.  In general, for Gaussian basis sets which are commonly used in chemistry, the asymptotic behavior of the constants are difficult to assess.  The worst-case scenario corresponds to the case where the Hamiltonian coefficients as well as the dipole moment coefficients have no pattern.  In this case, a lookup table is the most natural way to load these elements in.  Using a QROAM-based strategy, the work of~\cite{lee2021even} shows that this can be attained with an $\widetilde{O} \lb N^4 \rb$ scaling for the number of non-Clifford operations needed to implement the Hamiltonian $\mathtt{PREPARE}$ circuit and an $\widetilde{O}\lb N^2 \rb$ scaling for the dipole operator where $N$ denotes the number of basis functions considered.  Using similar tricks, the cost of the $\mathtt{SELECT}$ circuit using a Jordan-Wigner representation for the operators can be shown to be sub-dominant to the cost of implementing $\mathtt{PREPARE}$~\cite{lee2021even}.  These observations combined lead us to the conclusion that these simulations are indeed efficient provided that a polynomially large basis set is used. 

The remaining question is: given knowledge of the expectation value of an operator for one molecular configuration, how far can we extrapolate from that configuration before we violate the limits of the above accuracy bound? Sufficient conditions for the uncertainty can be found in our context by using bounds derived from analytic expressions for the derivatives of the eigenvalues and eigenvectors of the Hamiltonian. Let us consider a parameterized path on Hamiltonians $H(s)$ with $\ket{\psi(s)}$ as the ground state at time $s \in [0,1]$ and corresponding eigenvalue $E_{\psi}(s)$. Furthermore, let $\gamma(s)$ be the spectral gap of $H(s)$.
Provided that $\gamma(s) > 0$ and $H(s)$ is differentiable for all $s\in[0,1]$, we have that
\begin{align}
    \|\ket{\psi(1)} - \ket{\psi(0)}\| &= \left\|\int_{0}^1 \partial_s \ket{\psi(s)} \mathrm{d}s\right\| \nonumber\\
    &= \left\|\int_0^{1} \sum_{\phi \ne \psi} \ket{\phi} \frac{\bra{\phi(s)} \dot{H}(s) \ket{\psi(s)}}{E_\psi(s) - E_\phi(s)} \mathrm{d}s\right\|\nonumber\\
    &\le \max_s \sqrt{\sum_{\phi\ne \psi} \frac{\bra{\psi(s)} \dot{H}(s) \ketbra{\phi(s)}{\phi(s)}\dot{H}(s) \ket{\psi(s)}}{(E_\phi(s) -E_\psi(s))^2}}\nonumber\\
    &\le   \sqrt{\frac{\max_s\|\dot{H}(s)\|^2}{\min_s \gamma(s)^2}}={\frac{\max_s\|\dot{H}(s)\|}{\min_s \gamma(s)}}.
\end{align}
Next, let us show how this can be used to bound the difference in the expectation values of the projectors between different configurations.
By the triangle inequality, we have that
\begin{equation}
    | \bra{\psi(0)} \Pi_j \ket{\psi(0)} - \bra{\psi(1)} \Pi_j \ket{\psi(1)} | \leq | \bra{\psi(0)} \Pi_j \ket{\psi(0)} - \bra{\psi(0)} \Pi_j \ket{\psi(1)} | + | \bra{\psi(0)} \Pi_j \ket{\psi(1)} - \bra{\psi(1)} \Pi_j \ket{\psi(1)} |.
\end{equation}
Note that
\begin{equation}
    \begin{split}
        | \bra{\psi(0)} \Pi_j \ket{\psi(0)} - \bra{\psi(0)} \Pi_j \ket{\psi(1)} | &= \norm{\ketbra{\psi(0)}{\psi(0)} \Pi_j \ket{\psi_0} - \ketbra{\psi(0)}{\psi(0)} \Pi_j \ket{\psi(1)}} \\
        &\leq \norm{\ket{\psi(0)} - \ket{\psi(1)}} \leq {\frac{\max_s\|\dot{H}(s)\|}{\min_s \gamma(s)}}.
    \end{split}
\label{bound_diff_states}
\end{equation}
The same holds for the second term, $| \bra{\psi(0)} \Pi_j \ket{\psi(1)} - \bra{\psi(1)} \Pi_j \ket{\psi(1)} |$. This implies that
\begin{equation}
    | \bra{\psi(0)} \Pi_j \ket{\psi(0)} - \bra{\psi(1)} \Pi_j \ket{\psi(1)} | \le {2\frac{\max_s\|\dot{H}(s)\|}{\min_s \gamma(s)}}.
\label{extrp_bound}
\end{equation}
Recall from Lemma~\ref{lem:D} that we effectively split the projector decomposition of $A$ into two disjoint sets. The first set consists of projectors $\Pi_i$ for which we have $\widetilde{\epsilon}$-precise estimates $C_i$ of their expectation values. Importantly, these estimates need to be sufficiently precise so that we do not need a quantum computer to refine them. For simplicity, we assume here that we have perfect knowledge of all the expectation values at $s=0$. The other set contains the projectors whose expectation values are small enough such that their estimates can be refined via AAE.
As before, let us use $k$ to index those projectors that are evaluated on a quantum computer and let $P_0(k)$ be an upper bound on $\bra{\psi(0)} \Pi_k \ket{\psi(0)}$. Furthermore, let $P_0'(k)$ be an upper bound on $\bra{\psi(1)} \Pi_k \ket{\psi(1)}$. Note that Eq.~\eqref{extrp_bound} implies that it suffices to have
\begin{equation}
    P_0'(k) \leq P_0(k) + {2\frac{\max_s\|\dot{H}(s)\|}{\min_s \gamma(s)}}.
\end{equation}
In order to apply Lemma~\ref{lem:D}, we need to have that $P_0'(k) \leq \sin^2(\pi/6) = 1/4$. 
This can be achieved by ensuring that
\begin{equation}
    {\max_s\|\dot{H}(s)\|} \leq 2 \min_s \gamma(s) \lb \frac{1}{4} - \max_k P_0(k) \rb.
\end{equation}
Furthermore, we require
\begin{equation}
    {\max_s\|\dot{H}(s)\|} \le  \frac{{\min_s \gamma(s)}\epsilon}{4\norm{\beta}_1} 
\end{equation}
for $\widetilde{\epsilon}$ to satisfy the assumptions of Lemma~\ref{lem:D}.
Overall, we thus require that
\begin{equation}
    {\max_s\|\dot{H}(s)\|} \leq \min \left\{ \frac{{\min_s \gamma(s)}\epsilon}{4\norm{\beta}_1}, \min_s \gamma(s) \lb \frac{1}{4} - \max_k P_0(k) \rb \right\}.
\end{equation}
From this we can see that there exists a radius of Hamiltonians within which the expectation value of an operator can be extrapolated at low cost.

\subsection{Energy estimation through gradient integration
\label{sec:newton-cotes}}

Traditionally, energy estimation within the Born-Oppenheimer approximation in quantum simulation has followed a standard paradigm.  Specifically, one begins by estimating the energy directly using phase estimation at a point or configuration of interest.  The previous AAE-based approach can, of course, be used to provide estimates of the energy based on a previously estimated value by measuring the expectation value of the Hamiltonian in a given quantum state.  Here, we use a different approach that attempts to better use prior knowledge by integrating a directional derivative of the energy over a path that connects a known point to an unknown point. This approach, in essence, constructs a series of derivative estimates that are combined together to form a single estimate of the energy using a numerical integration formula.  Further, under the assumption that the second derivatives of the energy are not large along the path, we can use the previous point that we estimated to accelerate our estimation of the next point as per the previous discussion.  However, in principle this approach also has value outside of the framework of amplified amplitude estimation because it is not strictly speaking necessary to evaluate the small expectation values of the projectors that arise in this approach.

Specifically, we construct our numerical integration formulas through a family of integrators known as Newton-Cotes formulas.
\begin{lemma}[Newton-Cotes formulas]
\label{lem:energy}
Let $\partial_x E(x)$ be the spatial derivative of a function $E$ and assume that $\partial_x E(x)$ is analytic on $[-7/3,13/3]$ and has an analytic continuation as a single valued regular function on the contour in the complex plane $\mathcal{C} = \{z: z=1+ 3e^{i\phi} + 3^{-1} e^{-i\phi}\}$ for $\phi \in [0,2\pi)$. The $(N+1)$-point Newton-Cotes formula for odd $N>0$ provides an approximation of the form
$$
E(1) - E(-1)= \sum_{k=0}^{N} \aleph_k \partial_x E|_{x=x_k} + \Delta_N
$$
where $x_k = -1 + 2k/N$, $\aleph_k = \int_{-1}^1 \frac{\prod_{i\ne k}(x-x_i)}{\prod_{i\ne k}(x_i-x_k)} \mathrm{d}x$ and
$$
|\Delta_N| \le \frac{5}{3}\max_{z\in \mathcal{C}} |\partial_z E(z) | \left(\frac{3}{4} \right)^{N+1}\qquad \text{and}\qquad \sum_{k=0}^N \aleph_k \le 2(N+1).
$$
\end{lemma}
\begin{proof}
From Theorem 1 of \cite{kambo1970error} we have that the error in the Newton-Cotes formula can be computed in terms of a contour integral. The contour is chosen to be an ellipse where the sum of the semi-axes is equal to $\rho \geq 1$. 
Here we take $\rho=3$, although other choices are possible and will lead to different tradeoffs in the error bounds between the norm of the integrand and the exponential dependence:
\begin{equation}
    |\Delta_N| \leq \frac{4}{3}\frac{\rho^2+1}{\rho^2-1} \max_{z \in \mathcal{C}} \|\partial_z E(z)\|\left(\frac{2\rho}{\rho^2-1}\right)^{N+1}\le \frac{5}{3} \max_{z \in \mathcal{C}} |\partial_z E(z)| \left(\frac{3}{4} \right)^{N+1}.
\end{equation}
For us to use these expressions in LCU circuits, we need to also bound the sum of the coefficients, $\aleph_k$. First, note that
\begin{equation}
    \int_{-1}^1 \frac{\prod_{i\ne k} (x-x_i)}{\prod_{i\ne k} (x_k-x_i)}\mathrm{d}x \le \int_{-1}^1 \left|\frac{\prod_{i\ne k} (x_k-x_i)}{\prod_{i\ne k} (x_k-x_i)}\right|\mathrm{d}x \le 2.
\end{equation}
Thus,
\begin{equation}
    \sum_{k=0}^N\aleph_k = \sum_{k=0}^N \int_{-1}^1 \frac{\prod_{i\ne k} (x-x_i)}{\prod_{i\ne k} (x_k-x_i)}\mathrm{d}x \le 2(N+1).\label{eq:alephBd}
\end{equation}
\end{proof}
The following result uses the above lemma to state a bound on the error that arises from the combination of using a finite number of interpolation points in the Newton-Cotes formula as well as the approximation error in the values of the derivatives.
\begin{corollary}[Error bound for Newton-Cotes]
\label{cor:error}
Under the assumptions of Lemma~\ref{lem:energy}, we further have that if for each $x_k$ we are given a real number $F(x_k)$ such that $|F(x_k) - \partial_x E|_{x=x_k}|\le \delta$, then 
\begin{equation}
    \left|\sum_{k=0}^N \aleph_k F(x_k) - [E(1) - E(-1)]\right| \le \epsilon
\end{equation}
if 
\begin{equation}
    N \ge \frac{1}{\log(4/3)}\log\left(\frac{10 \max_{z\in \mathcal{C}} |\partial_z E(z) | }{3\epsilon}\right) -1
\end{equation}
and
\begin{equation}
    \delta \le \log(4/3)\epsilon / \left(4\log\left(\frac{10 \max_{z\in \mathcal{C}} |\partial_z E(z) | }{3\epsilon}\right)\right).
\end{equation}
\end{corollary}
\begin{proof}
The proof follows immediately from Lemma~\ref{lem:energy}, using the triangle inequality and setting both sources of error to $\epsilon/2$.  Specifically,
\begin{equation}
    \frac{5}{3}\max_{z\in \mathcal{C}} |\partial_z E(z) | \left(\frac{3}{4} \right)^{N+1} = \epsilon/2 \implies N = \frac{1}{\log(4/3)} \log(10\max_{z\in \mathcal{C}} |\partial_z E(z) |/3\epsilon) -1.
\end{equation}
Next, from the triangle inequality, we have that the error from approximating each of the derivatives within error at most $\delta$ is from~\eqref{eq:alephBd}
\begin{equation}
    \left|\sum_{k=0}^N \aleph_k( F(x_k) - \partial_x E|_{x=x_k}) \right| \le 2(N+1) \delta \le\left( \frac{2}{\log(4/3)} \log(10\max_{z\in \mathcal{C}} |\partial_z E(z) |/3\epsilon) \right)\delta.
\end{equation}
The final result follows by setting this equal to $\epsilon/2$ and solving for $\delta$.
\end{proof}

Next, we will consider the problem of evaluating the energy difference between two nearby points using this strategy.  
We will show how we can integrate the derivative of a quantity, such as energy, using a Newton-Cotes formula in order to compute the difference in that quantity between two points.  

\begin{theorem}[Energy estimation with Newton-Cotes]
     Assume the following:
    \begin{enumerate}
        \item Let $H: [-1,1] \mapsto \mathbb{C}^{2^n\times 2^n}$ be a mapping to Hermitian matrices such that $H(x) = \sum_{j} \alpha_j(x) U_j$ for unitary matrices $U_j$ and $\alpha_j(x) \in \mathbb{C}$. Furthermore, let $\gamma(x)$ be a lower bound on the spectral gap of $H(x)$.
        \item Let $E:\mathbb{C} \mapsto \mathbb{R}$ be a differentiable function such that $\Gamma := \max_{z\in \mathcal{C}}|\partial_z E(z)|$ and for any $z\in \mathbb{C}$ we have that the analytic continuation of the energy onto the complex plane is written as $E(z) = \bra{\psi_0(z)} H(z) \ket{\psi_0(z)}$ with $E(z)\in \mathbb{R}~\forall~z\in [-1,1]$.
        \item Let $G = \sum_{j=0}^{J-1} A_j = \sum_{j=0}^{J-1} \sum_{k=1}^{n_j} \beta_{jk} \Pi_{jk}$
        be the directional derivative operator for the energy such that for parameter $x$, $\partial_x E(x) = \bra{\psi_0(x)} G \ket{\psi_0(x)}$ at configuration $x$ and $\ket{\psi_0(x)}$ is an eigenstate of $H(x)$ for any $x\in [-1,1]$.  
        \item The assumptions of Corollary~\ref{cor:large_val} and Lemma~\ref{lem:Ge} hold for $\alpha(x)$ and $\gamma(x)$ for all $x\in [-1,1]$.
    \end{enumerate}
    Then there exists a quantum algorithm that takes the precise values of $E(-1)$ and $\bra{\psi_0(-1)} \Pi_{jk} \ket{\psi_0(-1)}$ for all $j,k$ and outputs an $\epsilon$-approximate estimate of $E(1)$ with probability of failure
    \begin{equation}
        P_f = \epsilon_\psi^2\in  e^{o(\max_x\|\alpha(x)\|_{1}/\min_x\gamma(x))}
    \end{equation}
    using a number of queries that scales as
    $$
        \widetilde{O}\left( \frac{\log(\Gamma)\max_x \|\alpha(x)\|_1\sum_j\max_x\sqrt{P_0(j;x)} \max_x\|\beta(x)\|_{1,1/2} \log(J/\delta')}{\min_x \gamma(x)\min_x |a_0(x)|\epsilon } \max_{j,x}\left(\frac{P_0(j;x)}{|\delta_j(x)|} \right) \right).
    $$
\label{thm:energy-diff}
\end{theorem}
\begin{proof}
First, we have from Corollary~\ref{cor:error} that we can estimate the integral within error $\epsilon$ using $N+1$ points for 
\begin{equation}
    N \in O \lb \log \lb\Gamma/\epsilon \rb \rb
\end{equation} 
and such that each of the values that we require are computed within error $O(\epsilon /\log(\Gamma/\epsilon))$.  We can then learn these $N+1$ values within the required error tolerance under the assumption that we have perfect knowledge of each value of $\bra{\psi_0(-1)} \Pi_{jk} \ket{\psi_{0}(-1)}$.  Since we have assumed that the assumptions of Corollary~\ref{cor:large_val} hold for all $x\in [-1,1]$, we can apply the corollary to estimate the expectation value of the derivative at each of the $N$ points in the interval since all of them are evaluated over $[-1,1]$.  Thus, by Corollary~\ref{cor:large_val}, we can compute a value $\mathcal{D}$ such that
\begin{equation}
    \left|\mathcal{D} - \int_{-1}^1 \bra{\psi_0(x)}G(x) \ket{\psi_0(x)} \mathrm{d}x\right| \le \epsilon.
\end{equation}
Denoting the cost of preparing the ground state as ${\rm Cost}(PE)$, we can achieve this using a number of queries that scales at most as
\begin{align}
    N_{queries} \in  \widetilde{O}\left( \frac{N{\rm Cost}(PE)\sum_j\max_x\sqrt{P_0(j;x)} \max_x\|\beta(x)\|_{1,1/2} \log(J/\delta')}{\epsilon } \max_{j,x}\left(\frac{P_0(j;x)}{|\delta_j(x)|} \right) \right)\nonumber\\
    \subseteq \widetilde{O}\left( \frac{\log(\Gamma){\rm Cost}(PE)\sum_j\max_x\sqrt{P_0(j;x)} \max_x\|\beta(x)\|_{1,1/2} \log(J/\delta')}{\epsilon } \max_{j,x}\left(\frac{P_0(j;x)}{|\delta_j(x)|} \right) \right).
\label{eq:tempCost}
\end{align}
As $H = \sum_j \alpha_j U_j$ and the spectral gap of the eigenstate in question is at least $\gamma(x)$, the cost of performing phase estimation to project onto the correct energy eigenvalue is given by Lemma~\ref{lem:Ge} to be
\begin{equation}
    {\rm Cost}(PE) \in \widetilde{O}\left(\frac{\max_x\|\alpha(x)\|_1}{\min_x|a_0(x)|\min_x \gamma(x)} \right),
\label{eq:PEcost2}
\end{equation}
assuming that the acceptable probability of failure obeys $\epsilon_{\psi}^2\in e^{o(\max_x\|\alpha(x)\|_{1}/\min_x\gamma(x))}$.  
By substituting~\eqref{eq:PEcost2} into~\eqref{eq:tempCost} we find that the required number of queries to $H$ and $G$ is in
\begin{equation}
    \widetilde{O}\left( \frac{\log(\Gamma)\max_x \|\alpha(x)\|_1\sum_j\max_x\sqrt{P_0(j;x)} \max_x\|\beta(x)\|_{1,1/2} \log(J/\delta')}{\min_x \gamma(x)\min_x |a_0(x)|\epsilon } \max_{j,x}\left(\frac{P_0(j;x)}{|\delta_j(x)|} \right) \right).
\end{equation}

All that remains to show is that the estimate of $\mathcal{D}$ provides us an $\epsilon-$approximation of the energy $E(1) = \bra{\psi_0(1)} H(1) \ket{\psi_0(1)}$ at $x=1$. Using the Hellmann-Feynman theorem and the fundamental theorem of calculus, we find that
\begin{equation}
\begin{split}
     \bra{\psi_0(1)} H(1) \ket{\psi_0(1)} &= \bra{\psi_0(-1)} H(-1) \ket{\psi_0(-1)} + \int_{-1}^1 \Big( \partial_x\bra{\psi_0(x)} H(x) \ket{\psi_0(x)} \Big) \mathrm{d}x \\
     &= \bra{\psi_0(-1)} H(-1) \ket{\psi_0(-1)} + \int_{-1}^1 \bra{\psi_0(x)} \big( \partial_x H(x) \big) \ket{\psi_0(x)} \mathrm{d}x \\
     &= \bra{\psi_0(-1)} H(-1) \ket{\psi_0(-1)} + \int_{-1}^1 \bra{\psi_0(x)} G(x) \ket{\psi_0(x)} \mathrm{d}x.
\end{split}
\end{equation}
From the triangle inequality we then have that
\begin{equation*}
\begin{split}
    &\Big| \bra{\psi_0(-1)} H(-1) \ket{\psi_0(-1)} + \mathcal{D} - \bra{\psi_0(1)} H(1) \ket{\psi_0(1)} \Big| \\
    &\leq \Big| \bra{\psi_0(-1)} H(-1) \ket{\psi_0(-1)} + \int_{-1}^1 \bra{\psi_0(x)} G(x) \ket{\psi_0(x)} \mathrm{d}x - \bra{\psi_0(1)} H(1) \ket{\psi_0(1)}\Big| \\
    &\quad + \Big| \mathcal{D} - \int_{-1}^1 \bra{\psi_0(x)} G(x) \ket{\psi_0(x)} \mathrm{d}x \Big| \leq \epsilon
\end{split}
\end{equation*}
and thus the claim follows.
\end{proof}
We see from this result that, provided the Hamiltonian is sufficiently smooth and gapped and if prior knowledge is provided, the asymptotic scaling can be better than the initial scaling.  In particular, there are two ways that this asymptotic scaling can be superior given that we are interested in finding the ground state energy of a Hamiltonian $H(1)$ that is close to the Hamiltonian $H(-1)$ of an initial configuration that we know well.
\begin{enumerate}
    \item If $\|G\|$ is small then the value of $\|\beta\|_{1,1/2}$ for the simulation will typically be small as well.
    \item If we have that $\|H(-1)-H(1)\|$ is small and the ground state is gapped, then we will have sufficient prior knowledge about $P_0(k)$ to estimate the expectation values of the individual projector terms accurately given that we know their values at $H(-1)$.
\end{enumerate}
This suggests that for each initial Hamiltonian there exists a ball of radius $\eta$ such that our method will be able to outperform traditional phase estimation as $\|\beta\|_{1,1/2}$ will shrink with the value of $\eta$ under the assumption that the ground state is gapped.  This is significant because it shows that if we have accurate prior knowledge about the energy and the expectation values of the projectors that comprise the directional derivative operator $G$, then amplified amplitude estimation can be used to provide a substantial advantage for estimating the values of the energy in a neighborhood about this point.  This shows that if we wish to learn a potential energy landscape for chemistry applications, we do not need to learn every point accurately in a mesh.  We can instead learn a few relevant points with great accuracy and then refine this prior knowledge to give cheaper estimates of the expectation values for nearby molecular configurations to those that we have already probed.

\section{Conclusion}
In this work, we provide an approach that we call amplified amplitude estimation (AAE) which uses amplitude amplification to increase the sensitivity of amplitude estimation.  
The main result is that we can estimate the expectation value of an operator within error $\epsilon$ using a number of queries that scale as
\begin{equation}
    N_\textrm{queries} \in O\left(\frac{P_0^{1/2}}{\epsilon}\left(\frac{P_0}{|\delta|}\right) \right),
\end{equation}
where $P_0$ is an upper bound on the expectation value and $\delta$ measures the minimum gap between our estimate and the true value of the estimated probability.  In cases where $|\delta|\in \Theta(P_0)$ and $P_0 \in \Theta(\epsilon)$ then we find $N_{\rm queries} \in O(1/\sqrt{\epsilon})$.  This shows that when estimating small probabilities, within error commensurate with the values that we aim to estimate, we can achieve scaling that is better than what we would expect from the Heisenberg limit.  Further, it is worth noting that this scaling does not simply arise from doing amplitude estimation on a small probability, despite the fact that $O(1/\epsilon)$ estimation is classically possible in the analogous limit.  We find instead that the way we use amplitude amplification needs to be adapted using prior knowledge in order to be able to improve upon the classical results in such cases and further generalize this to cases where we wish to estimate a sum of multiple expectation values.

We provide several ways in which amplified amplitude estimation can be used to provide an advantage for quantum simulation algorithms.  We show that this approach can provide an advantage for estimating expectation values of fermionic observables based on prior knowledge of the expectation values of the Pauli operators that compose the Jordan-Wigner decomposition of the observables.  We observe numerically, based on DFT and DMRG calculations, that existing classical methods can provide requisite levels of certainty about such expectation values to make such a procedure profitable.  We then conclude by discussing a new approach for estimating energy differences for molecules by integrating the derivative of the energy using Newton-Cotes formulas to perform the integration.  These methods can, given sufficiently strong prior knowledge and appropriate continuity assumptions of the Hamiltonian, provide improved scaling for energy estimation.

Looking forward, while our results show that prior information can be used to improve our ability to learn energies and expectation values, there are a number of directions to explore beyond it.  The first and most obvious question is whether these techniques can provide practical advantages for chemistry simulations.  Understanding the practical range where these effects provide a substantial advantage is necessary to understand the extent to which these ideas can improve upon classical or quantum estimates or cheaply probe the energy landscape of a system given a few points where the energy is known with high confidence.   Further, it may be interesting to investigate whether it is possible to apply this technique to cases where we aim to learn a vector of expectation values at cost that is sub-linear in the number of terms.  Analogous results are possible using quantum gradient estimation~\cite{huggins2021nearly} but in this setting it is not obvious how to use these techniques to improve the estimates given the multitude of prior estimates needed.

From a broader perspective, prior knowledge is still an under-utilized resource in quantum chemistry simulation algorithms and finding improved ways that perturbative estimates, interpolation techniques or better ways of incorporating prior knowledge in phase estimation procedures promises to meaningfully reduce the costs of simulation.  It is our hope that these ideas, and ones related to it, will be instrumental in helping useful applications of quantum chemistry simulation such as drug or catalyst design become a reality.

\section*{Acknowledgements}
SS acknowledges support from a Research Award from Google Inc., NSERC Discovery Grants, an Ontario Graduate Scholarship as well as support from Boehringer Ingelheim.  NW's work on this project was supported by the “Embedding Quantum Computing into Many-body
Frameworks for Strongly Correlated Molecular and Materials Systems” project, which is funded by the U.S. Department of Energy (DOE), Office of Science, Office of Basic Energy Sciences, the Division of Chemical Sciences, Geosciences, and Biosciences. The authors thank Clemens Utschig-Utschig for useful discussions and feedback. NW would also like to thank Robin Kothari for his insightful comments about classical estimation of small probabilities.

\bibliography{biblio}

\end{document}